\documentclass[journal,twoside,web]{ieeecolor}
\usepackage{generic}
\usepackage{cite}
\usepackage{amsmath,amssymb,amsfonts}
\usepackage{algorithmic}
\usepackage{graphicx}
\usepackage{textcomp}
\usepackage[ruled,vlined]{algorithm2e}

\usepackage[dvipsnames]{xcolor}
\usepackage{dsfont}
\usepackage{epsfig}
\usepackage{subfigure}
\usepackage{float,caption}
\usepackage{bm}
\usepackage{mathtools}
\usepackage{physics}
\usepackage{theoremref}

\DeclareMathOperator*{\argmax}{arg\,max}
\DeclareMathOperator*{\argmin}{arg\,min}

\DeclareMathOperator*{\sinri}{\rm SINR_i}

\DeclareMathOperator*{\sinrvi}{\text{${\rm SINR}_{v_i}$}}

\DeclareMathOperator*{\sinr}{\rm SINR}
\DeclareMathOperator*{\sinrthetai}{\text{${\rm SINR} (p,\theta_i)$}}

\newtheorem{theorem}{Theorem}
\newtheorem{lemma}{Lemma}
\newtheorem{corollary}{Corollary}

\newtheorem{proposition}{Proposition}
\newtheorem{remark}{Remark}

\newtheorem{example}{Example}
\newtheorem{definition}{Definition}
\newtheorem{assumption}{Assumption}
\newtheorem{problem}{Problem}

\DeclareMathAlphabet{\mathpzc}{OT1}{pzc}{m}{it}

\def\BibTeX{{\rm B\kern-.05em{\sc i\kern-.025em b}\kern-.08em
    T\kern-.1667em\lower.7ex\hbox{E}\kern-.125emX}}
\begin{document}
\title{Mean-Field Transmission Power Control in Dense Networks}

\author{Yuchi~Wu, Junfeng~Wu, Minyi~Huang, and Ling~Shi
\thanks{The work by Y.~Wu and L.~Shi is supported by a Hong Kong RGC General Research Fund 16204218.}
\thanks{The work by J.~Wu is funded by NSFC No. 62003303 from National Science Foundation of China.}
\thanks{The work by M.~Huang is supported in part by Natural Sciences and Engineering Research Council (NSERC) of Canada under a Discovery Grant.}
\thanks{Y.~Wu and L.~Shi are with the Department of Electronic and Computer Engineering, Hong Kong University of Science and Technology, Hong Kong (e-mail: ywubj@connect.ust.hk; eesling@ust.hk).}
\thanks{J.~Wu is with the College of Control Science and Engineering, Zhejiang University, Hangzhou 310027, China (e-mail: jfwu@zju.edu.cn).}
\thanks{M.~Huang is with the School of Mathematics and Statistics, Carleton University, Ottawa, ON K1S 5B6, Canada (e-mail: mhuang@math.carleton.ca).}
}

\maketitle

\begin{abstract}
We consider uplink power control in wireless communication when a large number of users compete over the channel resources. The CDMA protocol, as a supporting technology of 3G networks accommodating signal from different sources over the code domain, represents the orthogonal multiple access (OMA) techniques. With the development of 5G wireless networks, non-orthogonal multiple access (NOMA) is introduced to improve the efficiency of channel allocation. Our goal is to investigate whether the power-domain NOMA protocol can introduce performance improvement when the users interact with each other in a non-cooperative manner. It is compared with the CDMA protocol, where the fierce competition among users jeopardizes the efficiency of channel usage. In this work, we conduct analysis with an aggregative game model, and show the existence and uniqueness of an equilibrium strategy. Next, we adopt the social welfare of the population as the performance metric, which is the average utility achieved by the user population. It is shown that under the corresponding equilibrium strategies, NOMA outperforms CDMA by higher efficiency of channel access for uplink communications.
\end{abstract}


\section{Introduction}
\label{sec:introduction}
Power allocation has been conducted in wireless communications, wireless sensor networks as well as networked control systems. As a variant of resource allocation, it mainly deals with the tradeoff between the performance achieved and the power consumption. Moreover, recent advances in the fifth-generation (5G) communication network~\cite{shafi20175g} have led to a resurgence of interest in transmission power control.

Earlier works in wireless communications optimize the performance through appropriate energy allocation. Decentralized approaches are frequently investigated, among which game-theoretic methods are powerful tools for modeling non-cooperative channel access behaviors, especially for uplink users. Alpcan et al.~\cite{alpcan2002cdma} and Huang et al.~\cite{huang2006auction} considered a channel access game where each user strives for a better quality of service (QoS), and mechanism design is employed to improve the social effect. Moreover, centralized power allocation for maximizing the sum data rate is considered by Fischione et al.~\cite{fischione2009power}. Later on, to accommodate multiple users in uplink transmission more efficiently with successive interference cancellation, optimal power allocation has been investigated by Xu et al.~\cite{xu2013decentralized}.


Recent studies address transmission power allocation for a large number of users at the mean-field limits. Huang et al.~\cite{huang2003individual,huang2004uplink} investigated uplink power control of a large number of players with a mean-field model. Moreover, Semasinghe and Hossain~\cite{semasinghe2015downlink} dealt with downlink power allocation and compared the mean-field equilibrium performance under different utility functions. Typically, mean-field games consider dynamic models with an infinite number of players, where the impact from the opponents is modelled collectively as a mean-field term, as illustrated by Caines et al.~\cite{caines2017handbook}. As an alternative of game-theoretic model with mean-field type coupling, an aggregative game model~\cite{jensen2018aggregative} featuring static decision with either finite or infinite players is more appropriate for resource allocation in a wireless communication system.

Aggregative game has been widely adopted in wireless communications and networked games. For example, decentralized channel access among secondary users in a cognitive radio system is modelled by Pang et al.~\cite{pang2010design} as an aggregative game. A similar model is also employed for a large-scale channel access with incentive design at the base station, as investigated by Zhou et al.~\cite{zhou2017private}. These works adopt conventional mean-field modelling where each agent faces the same aggregate of the opponents' strategy. However, in a more recent work on networked games by Parise and Ozdaglar~\cite{parise2019variational}, where the mean-field aggregate faced by different agents is dependent on its neighbouring network topology, variational inequalities are employed for the equilibrium analysis.

Motivated by the advances of the 5G communication network, several recent works~\cite{zhu2017optimal,sun2018joint,zhang2018secure} investigated the optimal power allocation problem with non-orthogonal multiple access~(NOMA) employed among users sharing the same portion of the physical resource. The receivers under the NOMA protocol adopt successive interference cancellation~(SIC), which improves the chance of successful decoding at the receiver. According to Vaezi et al.~\cite{vaezi2018non}, the advance of computational capability of user equipment has enabled the deployments of NOMA for practical applications.

Our results are partially related to the existence and uniqueness of equilibrium strategy under CDMA, as investigated by Alpcan et al.~\cite{alpcan2002cdma} as well as Aziz and Caines~\cite{aziz2016mean}, and we have extended the equilibrium analysis to NOMA. The main difference between the modelling in this work and in the literature is two-fold. First, we employed a aggregative game model with mean-field limits for NOMA, which to the best of our knowledge is novel. Second, we have embedded the power control game under CDMA and NOMA in a unified framework, which enables a qualitative performance comparison. Following the direction of Xu and Cumanan~\cite{xu2017optimal} as well as Wei et al.~\cite{wei2019performance}, this work further discusses the social welfare comparison of uplink channel access between the conventional orthogonal multiple access schemes and NOMA. Past works~\cite{wei2018performance,chen2017optimization,zeng2017capacity} only investigated the performance comparison with centralized resource allocation, where all the users have a common objective function and their channel access behaviors are dictated by the base station. In this paper, we consider a population of selfish users accessing the shared wireless channel with a non-cooperative game model. This is suitable when the channel users are interested in pursuing their individual objectives and are accommodated by a micro base station in the 5G network.

When modeling the user population behaviors under the NOMA protocol, one faces several challenges:
\begin{itemize}
\item[(1)] In the game, channel gain at each user varies. Hence, it is crucial to tackle the complexity arising from characterizing individual behaviors in presence of \textbf{a large population};
\item[(2)] Though the equilibrium strategy for power control game under CDMA has been well analysed~\cite{alpcan2002cdma,aziz2016mean}, there is no analysis on the existence and uniqueness of power control strategy under NOMA protocol. Due to the asymmetric pattern of the decoding algorithm under NOMA, the equilibrium analysis for CDMA cannot be directly extended to NOMA;
\item[(3)] The specific \textbf{channel fading model} for uplink users is \textbf{not always available} since the physical environment where the users are located is in general challenging to describe. Thus, it is desirable to have general results not depending on the specific form of channel distribution;
\item[(4)] The \textbf{social welfare comparison} between two communication protocols is in general challenging in a non-cooperative game, as the \textbf{trend of changes} in equilibrium strategies under different protocols are \textbf{difficult to evaluate}.
\end{itemize}

We will address these challenges through appropriate system modeling and perturbation-based approaches. The contributions of this paper are listed below:
\begin{itemize}
\item[(1)] We model the interactions among a large number of players under CDMA and NOMA as \textbf{aggregative games}, where the opponents' action is modelled as a mean-field effect (i.e., the interference) for a large population. The \textbf{equilibria} under CDMA and NOMA are characterized in~\textbf{\thref{THM:CDMA_contraction_mapping}} and~\textbf{\thref{THM:NOMA2_contraction_mapping}}. To the best of our knowledge, such an aggregative game framework with mean-field limit modelling for NOMA (with non-identical mean-field interferences faced by different users) is new, and the existence and uniqueness of \textbf{equilibrium} power control strategy under \textbf{NOMA} has not been addressed in the literature. The results we have obtained do \textbf{not depend} on a specific \textbf{channel fading model}, i.e., the distribution of the channel gain, and can be applied to a wide range of physical environments;
\item[(2)] Through establishing a contraction property of the strategy update at each user in the population, a \textbf{distributed algorithm} (\textbf{Algorithm~\ref{algo:distributed_NE}}) is proposed for CDMA (\textbf{\thref{THM:CDMA_contraction_mapping}}) and NOMA (\textbf{\thref{THM:NOMA2_contraction_mapping}}) such that the strategy profile of the user population is guaranteed to converge to the unique equilibrium strategy;
\item[(3)] We compare the social welfare at the equilibria between CDMA and NOMA. Due to the difficulties in calculating the equilibrium strategies and directly evaluating the corresponding values of the social welfare, we propose a \textbf{perturbation-based} approach to find the trend of change in the equilibrium performance from CDMA to NOMA. The \textbf{social welfare} achieved under the equilibrium of power control game with NOMA \textbf{strictly dominates} that with CDMA (\textbf{\thref{lemma:nec_optimal}}).
\end{itemize}

The remainder of this paper is organized as follows. The preliminaries on wireless communications are introduced in Section~\ref{sec:preliminaries}. In Section~\ref{sec:MFG}, we formulate the interaction between channel users as an aggregative game. Section~\ref{sec:CDMA} and Section~\ref{sec:NOMA} characterize the equilibrium strategies under CDMA and NOMA protocols, respectively. In Section~\ref{sec:performance_comparison}, we compare the social welfare at the equilibria. Section~\ref{sec:individual_properties} studies the individual behaviors at the equilibrium. To illustrate the social welfare comparison results, we present the numerical simulations in Section~\ref{sec:simulations}. Section~\ref{sec:conclusions} concludes the paper.

To highlight the structure and contribution of this paper, a flow chart is displayed in Fig.~\ref{fig:logic_flow}, where the comparison between CDMA and NOMA forms a key contribution.

\begin{figure}[H]
	\begin{center}
		\includegraphics[width=0.5\textwidth]{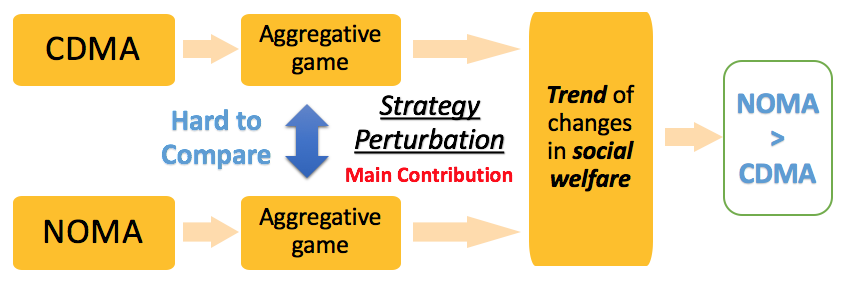}	
		\caption{The logic flow and main contribution of this paper.}
		\label{fig:logic_flow}
	\end{center}
\end{figure}

\subsection*{Notations:}
We denote the set of non-negative real numbers by $\mathbb{R}_+$, the set of positive real numbers by $\mathbb{R}_{++}$, and the set of non-negative integers by $\mathbb{N}_+$. For any Lebesgue measurable set~$A \subset \mathbb{R}$, denote its Lebesgue measure by $\lambda(A)$. The abbreviation ``a.e." is adopted for ``almost everywhere". For any $l \leq u$, define a truncation operator as $[x]_l^u:=\min \{u, \ max\{l,x\} \}$ for $x \in \mathbb{R}$. The modulus of a complex number~$z = x+ iy \in \mathbb{C}$ is~$\norm{z}:=\sqrt{x^2+y^2}$, where~$x,y \in \mathbb{R}$.

\section{Preliminaries on uplink wireless communication}
\label{sec:preliminaries}
In this paper, we consider a transmission power allocation problem that emerges in wireless networks where multiple signal sources are users of a single wireless communication channel. Different protocols are considered at the physical layer (PHY) of the communication channel, and the received signals are decoded by the base station. The realized data rate, i.e., the maximum rate of reliable communication supported by this channel, is the performance metric. The block diagram of the problem is presented in Fig.~\ref{fig:system}, and the details of each component are elaborated below.
\begin{figure}[H]
	\begin{center}
		\includegraphics[width=0.5\textwidth]{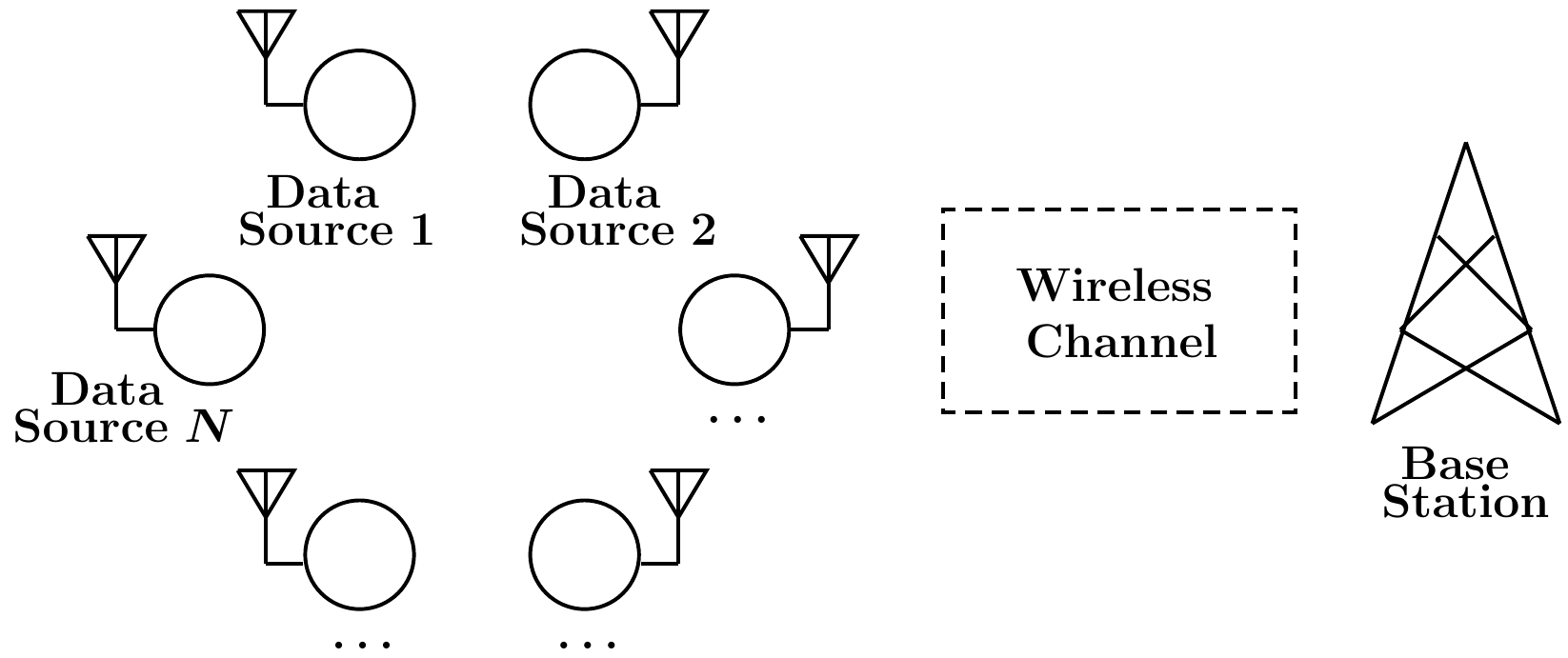}	
		\caption{Uplink channel access competitions of multiple users over a shared wireless channel to the base station.}
		\label{fig:system}
	\end{center}
\end{figure}

\subsection{Data sources}
There are~$N$ independent data sources as selfish users of the wireless channel. The information provided by each data source is modelled as a continuous-time waveform $x_i (t) \ (t \in \mathbb{R}_+)$, where $i \in \mathcal{N} := \{1,2,\ldots,N\}$. Assume that each of the signal source $x_i(t)$ has a unit average power level throughout the time horizon, i.e., $\lim\limits_{T \to \infty} \frac{1}{T} \int\limits_0^T \mathbb{E} [\norm{x_i(t)}^2] dt = 1$ holds for any~$i \in \mathcal{N}$. This is for notational convenience in describing the power level of the transmitted signal at the antennas.

Without loss of generality, we assume that each of the original continuous-time signal has a bandwidth bounded by~$B$, which is considered as a physical limitation of the communication channel. Based on Nyquist sampling theorem, we can always find a common sampling period $\Delta t$ for all the data sources such that the inequality $\frac{1}{\Delta t} > 2 B$ holds, and the sampled signal under this sampling rate carries full information of the original signal sources. For transmission through digital devices, a sampled version $x_i [k] \ (k \in \mathbb{N}_+)$ of the original signal $x_i(t) \ (t \in \mathbb{R}_+)$ is generated. Hence, we obtain the discrete-time signal sources $x_i [k]:=x_i (k \cdot \Delta t)$ for each $i \in \mathcal{N}$. Since each signal source contains different amount of information, they feature different data rates. We denote the data rate of source $i$ as $R_i \ \rm bits/s$ ($R_i > 0$).

\subsection{Gaussian channel and transmission antennas}
The communication channel between the data sources and the base station can be modelled as an additive white Gaussian noise (AWGN) channel as follows
\begin{equation}
y [k] = \sum\limits_{i=1}^N \sqrt{a_i} h_i x_i [k] + w [k],
\label{eq:AWGN}
\end{equation}
where $x_i [k]$ is the transmitted signal with transmission power level is $a_i$, $w[k] \sim \mathcal{CN} (0, N_0)$ is a complex Gaussian white noise process, and~$y [k]$ is the aggregate received signal at the base station. In addition, we assume time-invariant flat-fading character of the channel at each data source, and denote the uplink channel gain for each data source $i$ as a $\mathbb{C}$-valued random variable $h_i$. The probability density function of the squared magnitude $\norm{h_i}^2 \geq 0$ of the channel gain is~$f_i (x)$. Obviously, $f_i (x) \geq 0$ for any $x \in \mathbb{R}_+$ and $\int_{\mathbb{R}_+} f_i (x) dx = 1$.

In this paper, we assume that the channel gains $\norm{h_1}^2, \norm{h_2}^2, \ldots, \norm{h_N}^2$ are i.i.d. random variables, for which the probability density function is simply denoted as~$f$. Intuitively, all the local users are located within the coverage area of a base station, so that they share similar environments. For more generality, we do not take a specific form of distribution for the channel gain $h_i$. However, we require the following assumption to facilitate the analysis in this paper.

\begin{assumption}
The squared norm of the channel gain has a finite first-order moment: $\mathbb{E} [\norm{h_i}^2]:= \int_{\mathbb{R}} x f (x) dx <\infty$.
\thlabel{assumption:channel_gain}
\end{assumption}

This assumption enables the adoption of the strong law of large numbers (SLLN) when evaluating the interference from other transmitters.

\begin{example}
Rayleigh fading is commonly adopted to characterize the overall fading effect when a large number of reflectors and scatters present in the path, see~\cite{tse2005fundamentals}. The squared magnitude~$\norm{h_i}^2$ of the channel gain has the following probability density function
\begin{equation}
f (x) =
\begin{cases}
\frac{1}{\sigma^2} \exp \left( - \frac{x}{\sigma^2} \right), & x \geq 0 ; \cr
0, & \text{Otherwise},
\end{cases}
\label{eq:Raileigh_square}
\end{equation}
where the parameter~$\sigma > 0$. \thref{assumption:channel_gain} is clearly satisfied in this case.\\
\end{example}

The channel state information (CSI) $h_i$ is available at the transmitter $i$. This is easily achieved by sending a pilot signal prior to the transmission stage. At the receiver, the quality of a certain signal source in the received waveform~$y[k]$ is characterized by its~$\sinr$. To improve the $\sinr$ of a signal $x_i [k]$ for a higher probability of successful decoding, a possible method is to raise the transmission power level. For this purpose, the antenna at each data source $i$ will choose an appropriate power level $a_i$. Reliable decoding can be realized subject to good~$\sinr$ condition, see Section~\ref{subsec:BS}.

Through transmitting signal from different sources at different power levels, a novel multiple access strategy is implemented on the Gaussian channel (\ref{eq:AWGN}). This strategy exploits the power domain within the same physical channel resource, i.e., time, frequency, etc. Different from the conventional orthogonal multiple access protocols such as time-division multiple access (TDMA) and frequency-division multiple access (FDMA), the multiple access over the power domain is an example of non-orthogonal multiple access (NOMA), one of the supporting technologies in 5G networks~\cite{dai2015non}.

\subsection{Base station}
\label{subsec:BS}
A base station is located at the end of the uplink Gaussian channel, which is capable of sensing the channel gains of all communication links. The received signal from each data source~$i$ will be processed independently at the base station. The specific physical model adopted for describing the fundamental limits in signal decoding as well as two different decoding strategies is introduced below.

The data rate of source $i$ is $R_i > 0$. Based on Shannon's theorem~\cite{shannon1948mathematical}, in order to reliably decode the signal from data source $i$ with decoding error going to zero asymptotically as the block length increases, it is necessary for the instantaneous channel capacity of this link to exceed the corresponding data rate $R_i$~\cite{tse2005fundamentals} as follows:
\begin{equation}
\log_2 \left( 1 + \norm{h_i}^2 \sinri \right) > R_i.
\label{ineq:Shannon}
\end{equation}

Thus, the indicator for the successful decoding of the signal from data source $i$ is defined as
\begin{equation}
\gamma^{(i)} :=
\begin{cases}
1, & \sinri > \frac{2^{R_i} - 1}{\norm{h_i}^2}; \cr
0, & \text{Otherwise}.
\end{cases}
\label{eq:indicator}
\end{equation}

With the model for signal decoding as well as related notations introduced, the distribution for the decoding outcomes $\gamma^{(i)}$ relies on the distribution of~$\norm{h_i}$, and is determined by the outage probability
\begin{equation}
p_{out}^{(i)} (R_i) := \mathbb{P} \bigg \{ \log_2 \left( 1 + \norm{h_i}^2 \sinri \right) < R_i \bigg \}.
\label{eq:outage_prob}
\end{equation}

Accordingly, the probability of a packet arrival can be expressed as
\begin{equation}
\mathbb{P} \{\gamma^{(i)} = 1 \}  = 1 - p_{out}^{(i)} (R_i).
\label{eq:prob_arrival}
\end{equation}

As the users~$i \in \mathcal{N}$ are interested in supporting the success in decoding their source signal at the base station, each of them will attempt to achieve a higher instantaneous channel capacity  in order to accommodate its data waveform at the shared wireless medium.

For a Gaussian channel, there are mainly two decoding strategies when multiple signal sources are transmitting through their uplink channels: single user decoding~(SUD) in CDMA and multi-packet reception (MPR) in NOMA. More details for the two methods can be found in Zhang and Haenggi~\cite{zhang2014performance}. The power allocation for the data sources will differ when different decoding strategies are employed. Now we briefly explain the ideas of SUD in CDMA and MPR in NOMA. Details of these protocols and the corresponding decoding algorithms at the receiver will be elaborated in Sections~\ref{sec:CDMA} and~\ref{sec:NOMA}.

\subsubsection{Single user decoding (SUD) in CDMA}
For multiple data sources, a communication system employing CDMA deals with each communication link separately. While decoding the signal $x_i [k]$ from data source $i$, the signals from all the other sources $j \neq i$ are considered as interference. One typical application of this decoding strategy is the direct-sequence code division multiple access (DS-CDMA) system, where multiple users share a single channel. In order to mitigate the interference between different users, they adopt a group of signature sequences $\{\bm{s}_i\}_{i=1}^N$ such that they have high auto-correlation and low cross-correlation among themselves. Therefore, by selecting appropriate pseudo-noise (PN) sequences, it is possible to significantly reduce the interference of the received signal at the base station. Hence, the $\sinr$ of each user can maintain a certain level while being decoded.

Similar to FDMA and TDMA, the CDMA protocol belongs to the category of orthogonal multiple access (OMA), where different users are allocated different portions of the same type of physical resources, such as frequency bands, time slots, code chips, etc. On the other hand, when multiple users can share the same portion of physical resources, the protocol belongs to the category of non-orthogonal multiple access (NOMA), as introduced below.

\subsubsection{Multi-packet reception (MPR) in NOMA}
NOMA features resource sharing among different users such that a tradeoff between efficiency of the communication channel and the fairness among users is maintained. Multi-packet reception (MPR) decoding scheme is performed when NOMA is employed at the communication channel. In order to improve spectral efficiency and the user fairness~\cite{vaezi2018non}, NOMA is widely adopted in the fifth-generation (5G) communication network. In this paper, we analyse the effectiveness of power-domain NOMA. To facilitate the comparison of the two decoding methods, the same group of signature sequence~$\{\bm{s}_j\}_{j=1}^N$ as in the CDMA case is assigned to each user for implementation of the spread spectrum technique.

In NOMA, one of the key techniques is successive interference cancellation (SIC)~\cite{xu2013decentralized}, which is a recursive decoding algorithm. While the receiver is decoding the data from different signal sources following a certain order, the SIC algorithm will cancel the successfully decoded signals from the waveform before performing the subsequent decoding procedures, so that the overall probability of successful decoding is improved.

An example of SIC is the network-assisted interference cancellation and suppression (NAICS) user terminals adopted by 3GPP LTE-A standard~\cite{vaezi2018non}. The deployment of SIC at the receivers increases the time and space complexity. Fortunately, as shown by experiments, the complexity induced by adopting NOMA can be well accommodated by the capability of current user equipment~\cite{vaezi2019multiple}, which provided empirical evidence for the feasibility of the implementation of NOMA.

\section{Problem formulation: an aggregative game model for uplink power control}
\label{sec:MFG}

Each data source $i$ aims at uploading its local information to the base station, which can be considered as different users sharing a common computer network attempting to transmit its data to a distant server, such as cloud storage or cloud computing services.

Due to the selfish nature, each data source aims only at improving the chance of successfully decoding its own signal source at the base station while consuming less power.

The channel sharing behavior among these data sources can be modelled in a game-theoretic framework where different data sources are non-cooperative over the channel usage. To reduce the computational complexity for equilibrium analysis when a large number of players are involved, we adopt the idea of aggregative games~\cite{jensen2018aggregative}, where the impact of opponents' actions on each player is modelled as a collective effect. Then, each player is choosing its action as a response to this collective effect.

Now, we present the aggregative game model in detail.

\subsection{An aggregative game model}
In dense networks, a large number of users will attempt to access the Internet through a single access point~(AP). We employ a non-atomic game~\cite{bonneau2008non} model in this work, where the effect to the overall performance caused by a single player is negligible as the number of players~$N \to \infty$. Some notations are introduced below to adapt an aggregative game model to a large number of players, similar to Nekouei et al.~\cite{nekouei2017nash}.

In this paper, the interactive behavior among different users when competing over the uplink wireless channel is modelled as an aggregative game~$\mathcal{G} = \left(M, \mathcal{A}, u \right)$, of which each component is elaborated as follows.

\subsubsection{The set of players M}
We adopt symbolic representations for each type of users. With slight abuse of notations, the realization of the random variable~$h_i$ is denoted as~$h_i \in \mathbb{C}$ as well, which is the channel gain of user~$i$. Then, the identifier of this user is defined as $\theta_i = \norm{h_i}^2 \in \mathbb{R}_+$. We define the set of all possible identifiers as $M \subset \mathbb{R}_+$ such that any possible identifiers~$\theta_i$ belong to~$M$. To avoid triviality, we assume~$M \neq \emptyset$ and~$0 \notin M$. Thus, the set of all possible identifiers is a subset of positive real numbers, i.e., $M \subset \mathbb{R}_{++}$. Also, a Borel~$\sigma$-algebra is generated for the set $M$ and is denoted as $\mathcal{B}(M)$.


In practice, for players with different identifiers in a game with a large population, they will take up a certain ratio of presence in the population. We define a probability measure $P$ over the measurable space~$\left( M, \mathcal{B}(M) \right)$ to model the population of players, where the probability measure $P$ is induced by the probability density function $f(x)$ given in \thref{assumption:channel_gain}. To be specific, we have $P: \mathcal{B}(M) \to [0,1]$ such that for any $A \in \mathcal{B}(M)$, we have $P(A):=\int\limits_{x \in A} f(x) d \lambda(x)$ and $P(M)~=~1$. Hence, we choose $M \subset \mathbb{R}_{++}$ as the set of players.

\begin{assumption}
$f(x) > 0$ for any~$x \in M$, i.e., the probability density function~$f$ has positive value on the set of all possible identifiers~$M \subset \mathbb{R}_{++}$.
\thlabel{assumption:pdf_positive}
\end{assumption}

There is no loss of generality to consider~$f(x)>0$ for~$x \in M$. When this assumption is not satisfied, we can restrict the set of players to~$\tilde{M}=\{\theta_i \in M: f(\theta_i)>0\} \subset M$. All the subsequent analysis are performed with $\tilde{M}$. By definition of the probability measure $P$, it holds that $P(\tilde{M}) = P(M)=1$.

\begin{remark}
The identifiers of different users can be interpreted as sampled independently from the same distribution~$f$. Then, by the i.i.d. assumption on~$\{\theta_i, i \geq 1\}$, the empirical distribution of different types of players will converge to~$f$ if the number of players~$N \to \infty$. Hence, characterizing the ratio of players with different identities with a probability density function~$f(x)$ is feasible and theoretically sound.
\end{remark}

\subsubsection{The strategy space $\mathcal{A}$}
Now we define the strategy space of all players in the game.

Since the transmission antennas are controlled by analog circuits, the maximum transmission power is bounded. We define the feasible set of transmission power levels as a compact and convex set~$\mathcal{E}:=~[E_{\min},E_{\max}] \ \ (0 \leq E_{\min} < E_{\max} < \infty)$ without loss of generality.

As the number of players approaches infinity, and all the players with the same identifier are considered to be homogeneous, we denote the strategy of all the players in $M$ as a mapping $p: M \to \mathcal{E}$ such that the action chosen by the player $\theta_i \in M$ is $a_i := p(\theta_i) \in \mathcal{E}$.

Before the formal definition of the feasible strategy space of the players in $M$ is given, we introduce a new measure with which the aggregate effect of the players' strategies can be evaluated. For $M \subset \mathbb{R}_{++}$, a measure space is defined as $\left( M, \mathcal{B}(M), \lambda \right)$ based on the Lebesgue measure~$\lambda$. We define a new measure~$\nu$ as follows.
\begin{equation}
\nu (A) := \int_{A} w(x) d \lambda(x), \ \ \forall A \in \mathcal{B}(M),
\label{eq:new_measure}
\end{equation}
where the weight function $w: M \to [0,\infty)$ is
\begin{equation}
w(x) := x f(x), \ \ x \in M.
\label{eq:weight}
\end{equation}

Consequently, a new measure space $\left( M, \mathcal{B}(M), \nu \right)$ is generated, based on which we define the feasible strategy space as the set of functions
\begin{align}
\mathcal{A}&:=\{p: (p: M \to \mathbb{R}) \ \& \ (p \text{ is } \nu-\text{a.e. } \mathcal{E} \text{-valued} ) \},
\label{eq:action_space}
\end{align}
where~$p$ is a Lebesgue measurable function. We adopt the convention that two functions are identified as the same element of~$\mathcal{A}$ if they are equal~$\nu$-a.e.

\begin{definition}
For any Lebesgue measurable function~$g: M \to \mathbb{R}$, we introduce the norm
\begin{equation}
\norm{g}_1^{\nu} := \int\limits_{x \in M} \abs{g(x)} d \nu (x) = \int\limits_{x \in M} \abs{g(x)} w(x) d \lambda(x).
\end{equation}
\thlabel{def:weighted_1norm}
\end{definition}

\begin{definition}
For a bounded Lebesgue measurable function $f: M \to \mathbb{R}$, we define its essential supremum norm based on the measure~$\nu$ as~$\norm{f}_{\infty}^{\nu} := \inf \{ C > 0: \abs{f(x)} \leq C \text{ $\nu$-a.e.} \}$.
\thlabel{def:ess_supremum_norm}
\end{definition}

\subsubsection{The utility function $u$}
The utility function of a player with identifier~$\theta_i \in M$ is denoted as~$u(a_i, p, \theta_i)$, where it chooses an action~$a_i \in \mathcal{E}$ in response to~$p \in \mathcal{A}$, which is adopted by the opponents. Its utility function takes a tradeoff between the achieved data rate and energy consumptions, which has the form
\begin{equation}
u (a_i , p, \theta_i) = \log_2 \left( 1 + \theta_i \cdot \sinrthetai \right) - \beta_i a_i, \ \ \forall \theta_i \in M,
\label{def:utility}
\end{equation}
where $\beta_i>0$ is the power penalty parameter at player~$i$. As CDMA and NOMA adopt different decoding algorithms, for a fixed strategy profile~$p \in \mathcal{A}$, the resulting $\sinr$ at the receiver's side will have different values under these two protocols.

\begin{assumption}
The power penalty parameter is identical for each uplink user~$i \in \mathcal{N}$, i.e., $\beta_1=\beta_2=\cdots=\beta_N=\beta>0$.
\thlabel{assumption:identical_weight}
\end{assumption}

We will adopt this common power penalty parameter~$\beta$ for all the users in the following analysis.

\begin{remark}
The assumption for identical power penalty parameter~$\beta_i$ in the utility function indicates the cost for unit power consumption is the same for every uplink channel user. Intuitively, this adapts to the case where the power budget available to each user is obtained from the same type of power source, e.g., batteries. It should be noted that with the utility function given in its current form~(\ref{def:utility}), nonidentical power penalty parameter~$\beta_i$ for different type~$\theta_i \in M$ of users will not affect the main results as the existence of equilibrium strategies and the social welfare comparison results do not depend on the specific value of~$\beta_i$ associated with a certain user type~$\theta_i \in M$, though the monotonicity of the equilibrium strategy under CDMA as indicated by Corollary~1 may not hold anymore.
\end{remark}

\subsubsection{The information set $\ \mathcal{I}$}
Though the realization of the channel gain~$h_i$ is only disclosed to the user itself, the probability density function (PDF)~$f(x)$ of the user identity~$\theta_i = \norm{h_i}^2$ in the population is common knowledge. The realized value of the channel gain for each user can be interpreted as sampled independently from this distribution~$f(x)$. This knowledge sharing is achieved as the base station collects statistics of the realized channel magnitude $\theta_i = \norm{h_i}^2$, and then informs each participating player of this distribution (e.g., through a quantized version of the PDF).

For user~$\theta_i \in M$, the available information set is
\begin{equation}
\mathcal{I}(\theta_i)=\{ \theta_i, N_0 \}, \ \theta_i \in M, 
\label{def:info_set}
\end{equation}
where~$N_0$ is the power spectrum density of the Gaussian noise in the channel, and~$f(\theta)$ is the probability density function of the user identity. Thus, user $\theta_i$ will choose a power level from~$\mathcal{E}$ based on the locally available information set~$\mathcal{I}(\theta_i)$.

\begin{remark}
To facilitate the analysis of the collective effects of an infinite number of users, it is also of interest to explain the almost sure convergence of the aggregative effects when a large number of users are considered. For a group of countable number of players $\{\theta_1, \theta_2, \ldots\}$, if the power control strategy profile~$p(\cdot)$ is fixed, the interference when CDMA is employed can be expressed by a new random variable $a_i \theta_i = p(\theta_i) \theta_i$. Its first-order moment is finite, i.e.,~$\mathbb{E} [p(\theta_i) \theta_i] \leq \mathbb{E} [\norm{p}_{\infty} \cdot \norm{h_i}^2] = E_{\max} \cdot \mathbb{E} [\norm{h_i}^2] = E_{\max} \cdot \overline{h}_i^2 < \infty$, hence Kolmogorov's strong law of large numbers (SLLN)~\cite{durrett2010probability} applies when analyzing the aggregative effects. The applicability of SLLN to NOMA can be justified in a similar manner.
\end{remark}

\subsection{Solution concept for a game}
Assume all players are rational in the sense that each of them will seek to optimize its own utility function (i.e., self-interested). In addition, the rationality of the players' behaviors is common knowledge among all players participating in the game $\mathcal{G}$, as illustrated by Gibbons~\cite{gibbons1992game}. Then, each player $i$ will choose a power level from $\mathcal{E}$ to optimize its utility based on its own information set. We call it a decision-making process if a player $\theta_i \in M$ predicts the strategies of other players based on the information available and chooses an action in response to the predicted actions.

Different from an optimization problem, in a game-theoretic setup, the decision-making process of each player $i$ is in the form of a ``best response" to the action profile of its opponents. Under this decision pattern, it is of interest to find an action profile of all players on which they agree. Thus, we proceed to introduce the concept of $\epsilon$-Nash equilibrium~\cite{huang2007large} for an aggregative game with a finite number of players.

\begin{definition}[$\epsilon$-Nash equilibrium]
A strategy profile $p^* \in \mathcal{A}$ is an $\epsilon$-Nash equilibrium of a game among $N \in \mathbb{N}_+$ players, if for the identity $\theta_i \in M$ of all the $N$ players, we have
\begin{equation}
u (p^*(\theta_i), p^*, \theta_i) \geq u (a_i, p^*, \theta_i) - \epsilon, \ \ \forall \ a_i \in \mathcal{E},
\end{equation}
where $\epsilon>0$.
\thlabel{def:NE}
\end{definition}

When the number of users~$N \to \infty$, the positive real number~$\epsilon \to 0$ according to Theorem~5.7 in Nourian et al.~\cite{nourian2012nash}. Hence, as the population size tends to infinity, the $\epsilon$-Nash equilibrium converges to a mean-field equilibrium defined as follows.

\begin{definition}[Mean-field equilibrium]
A strategy profile $p^* \in \mathcal{A}$ is a mean-field equilibrium of a game with an infinite number of players if for any $\theta_i \in M$, we have
\begin{equation}
u (p^*(\theta_i), p^*, \theta_i) \geq u (a_i, p^*, \theta_i), \ \ \forall \ a_i \in \mathcal{E}.
\label{eq:def_NE}
\end{equation}
\thlabel{def:MFE}
\end{definition}

For the aggregative game~$\mathcal{G}$, we define the best response of user $\theta_i \in M$ to the opponents' strategy $p \in \mathcal{A}$ as a set-valued mapping $\mathcal{BR}: M \times \mathcal{A} \to 2^{\mathcal{E}}$. For each $\theta_i \in M$,
\begin{equation}
\mathcal{BR}(\theta_i, p):=\bigg \{ a_i^* \in \mathcal{E} : u(a_i^*, p, \theta_i) \geq u(a_i, p, \theta_i), \ \forall a_i \in \mathcal{E} \bigg \}.
\label{def:BR}
\end{equation}

Therefore, a strategy $p^* \in \mathcal{A}$ is a mean-field equilibrium if and only if for any $\theta_i \in M$, $p^*(\theta_i) \in~\mathcal{BR} (\theta_i,p^*)$. Thus, it is a strategy profile $p^* \in \mathcal{A}$ that every user $\theta_i~\in~M$ adopts, and has no incentive to unilaterally deviate from. Next, the details are shown below for analyzing the existence and uniqueness of mean-field equilibrium strategies under CDMA and NOMA.

\section{CDMA transmission power game}
\label{sec:CDMA}
\subsection{Single user decoding  (SUD) in CDMA}
\subsubsection{Descriptions of CDMA protocol and SUD decoding algorithm}
As mentioned in the preliminaries, with the CDMA communication protocol, each user is allocated a unique signature sequence so that its transmitted signal can be spread over different sub-carriers (i.e., code chips) in order to mitigate the interference between different users.

As mentioned by~\cite{tse1999linear,ferrante2015spectral}, through appropriate selection of signature sequences of length $n_s$, the squared cross-correlation between the signature sequences $\bm{s}_k$ and $\bm{s}_j$ of users $k$ and $j$ ($k \neq j$) can be expressed as $\rho_{k,j} = (\bm{s}'_k \bm{s}_j)^2 \approx \frac{1}{n_s} = \frac{\alpha}{N}$, which is the gain of the interference induced by user $k$ to the received signal from user~$j$. Since each user in CDMA is assigned a distinct signature sequence, the length~$n_s$ satisfies~$n_s \geq N$, i.e.,~$0<\alpha \leq 1$. In practice, as~$\rho_{k,j} = \frac{1}{n_s}$ , it is preferred to have a larger length~$n_s$ in order to reduce the cross-correlation between signature sequences assigned to different users, and then limit the interference introduced. Therefore, the parameter~$\alpha$ is often chosen to satisfy~$0 < \alpha \ll 1$.

The $\sinr$ of the signal from data source~$i$ is expressed as
\begin{equation}
\sinri = \frac{a_i}{\sum\limits_{j \neq i} \rho_{j,i} a_j \norm{h_j}^2 + N_0} = \frac{a_i}{\sum\limits_{j \neq i} \frac{\alpha}{N} a_j \norm{h_j}^2 + N_0},
\end{equation}
for any~$i \in \{1,2,\ldots,N\}$. The receiver will attempt to decode the signal from each communication link independently. By Shannon's theorem and~(\ref{ineq:Shannon}), the outcome of decoding signal from source~$i$ is only dependent on~$\sinri$.

\subsubsection{The utility functions of the aggregative game using SUD}
When the population size~$N \to \infty$, if the user~$\theta_i \in M$ chooses~$a_i \in \mathcal{E}$ as its transmission power, its received~$\sinr$ is given as
\begin{align*}
\sinrthetai & = \lim\limits_{N \to \infty} \frac{a_i}{\frac{\alpha}{N} \sum\limits_{j \neq i}  a_j \theta_j + N_0}\\
&=\lim\limits_{N \to \infty} \frac{a_i}{\frac{\alpha \cdot (N-1)}{N} \cdot \frac{1}{N-1} \sum\limits_{j \neq i}  p(\theta_j) \theta_j + N_0}\\
& \xrightarrow[\text{SLLN}]{\text{a.s.}} \frac{a_i}{\alpha \mathbb{E}[p(\theta_j) \theta_j] + N_0}, \ \ \theta_i \in M,
\end{align*}
where almost sure convergence holds by~\thref{assumption:channel_gain} and Kolmogorov's strong law of large numbers.

Then the utility function of any user $\theta_i \in M$ in the case of CDMA can be expressed as
\begin{align}
u (a_i , p, \theta_i) &= \log_2 \left( 1 + \theta_i \cdot \sinrthetai \right) - \beta a_i \nonumber\\
&=\log_2 \left( 1 + \frac{\theta_i a_i}{\alpha \mathbb{E}[p(\theta_j) \theta_j] + N_0} \right) - \beta a_i.
\label{eq:utility_CDMA}
\end{align}

\subsection{Analysis of the CDMA transmission power equilibrium}
Before we attempt to find the best response of player~$\theta_i$, we first define a projection operator on~$\mathbb{R}$ according to Bertsekas and Tsitsiklis~\cite{bertsekas1989parallel}.
\begin{definition}
For any given closed interval $X \subset \mathbb{R}$, define an orthogonal projection operator $P_X: \mathbb{R} \to X$ such that
\begin{equation}
P_X (x) := \argmin\limits_{z \in X} \abs{z-x}, \ \ \forall x \in \mathbb{R}.
\end{equation}
\thlabel{def:projection}
\end{definition}

Given the utility function~(\ref{eq:utility_CDMA}) and a fixed strategy $p \in \mathcal{A}$ of the opponents $M \setminus \{\theta_i\}$, we can obtain the optimal action $a_i^* \in \mathcal{E}$ of the player $\theta_i$ based on the best response operator in~(\ref{def:BR}),
\begin{align}
a_i^* &\in \mathcal{BR}(\theta_i, p) \nonumber\\
&=\bigg\{ P_{\mathcal{E}} \left( \frac{1}{\beta \ln 2} - \frac{\alpha \mathbb{E}[p(\theta_j) \theta_j] + N_0}{\theta_i} \right) \bigg\},\ \ \theta_i \in M.
\label{eq:BR_CDMA}
\end{align}

As the utility function $u(a_i, p, \theta_i)$ of the player with fixed identifier $\theta_i$ and fixed opponents' strategy~$p$ is strictly concave with respect to the power control action $a_i \in \mathcal{E}$, there is a unique maximizer $a_i^*$ of the utility function, as indicated in Theorem 9.17 in Sundaram~\cite{sundaram1996first}. Accordingly, the set of the best response of any player~$\alpha _i \in M$ is a singleton.

Then the space of all strategy profiles that induced a bounded interference term~$\alpha \mathbb{E}[p(\theta_j) \theta_j]$ under the CDMA protocol is a vector space defined as
\begin{align}
L^1(M, \mathbb{R},\nu) &:= \{ p: (p: M \to \mathbb{R}) \ \& \ (\norm{p}_1^{\nu} < \infty) \}.
\label{eq:space_1norm}
\end{align}

\begin{remark}
It is well known that~$L^q$ space is a Banach space for any~$1 \leq q \leq \infty$~\cite{luenberger1997optimization}. The case of weighted norm under a change of measure has been shown in Fischer-Riesz theorem (Theorem~7.18 of~\cite{klenke2013probability}), which applies to our scenario, i.e.,~$\left(L^1(M, \mathbb{R},\nu), \norm{\cdot}_1^{\nu} \right)$ is a Banach space. Details of the proof are omitted.
\thlabel{remark:Lp}
\end{remark}

The operator to be defined below will perform the truncation of any strategies in the space $\left(L^1(M, \mathbb{R},\nu), \norm{\cdot}_1^{\nu} \right)$ to the set of feasible strategies $\mathcal{A}$, due to the power limitations of the circuits of the transmitter.

\begin{definition}
Given the set of feasible strategies~$\mathcal{A} \subset~L^1(M, \mathbb{R},\nu)$, 
 the truncation  operator is 
 the mapping $\mathcal{T}: L^1(M, \mathbb{R},\nu) \to \mathcal{A}$ such that for an arbitrary $p \in L^1(M, \mathbb{R},\nu)$, we have that the resulting strategy profile
\begin{equation}
\tilde{p} := \mathcal{T} (p),
\end{equation}
which satisfies~$\tilde{p} (x) = P_{\mathcal{E}} (p(x))$, $\forall x \in M$.
\thlabel{def:saturation}
\end{definition}

A property of the operator $\mathcal{T}$ is given in the lemma below.
\begin{lemma}
The operator $\mathcal{T}: L^1(M, \mathbb{R},\nu) \to L^1(M, \mathbb{R},\nu)$ is non-expansive, i.e., for two arbitrarily picked elements $p^{(1)}, p^{(2)} \in L^1(M, \mathbb{R},\nu)$, we have
\begin{equation}
\norm{\mathcal{T}\left(p^{(1)}\right) - \mathcal{T} \left( p^{(2)} \right) }_1^{\nu} \leq \norm{p^{(1)} - p^{(2)}}_1^{\nu}.
\label{ineq:nonexpansive}
\end{equation}
\thlabel{lemma:T_nonexpansive}
\end{lemma}

\begin{proof}
This lemma is a direct extension of the projection theorem in Euclidean space~\cite{bertsekas1989parallel}. For any $p^{(1)}, p^{(2)} \in L^1(M, \mathbb{R},\nu)$, we can obtain 
\begin{align*}
& \ \ \ \ \norm{\mathcal{T}\left( p^{(1)} \right) - \mathcal{T}\left( p^{(2)} \right)}_1^{\nu}\\
& = \int\limits_{x \in M} \abs{\mathcal{T}\left( p^{(1)} \right) (x) - \mathcal{T}\left( p^{(2)} \right) (x)} w(x) d \lambda(x).
\end{align*}
Moreover, based on \thref{def:saturation} as well as the non-expansive property of the projection operator $P_{\mathcal{E}}$ as given in the projection theorem (i.e., Proposition 3.2 by Bertsekas and Tsitsiklis~\cite{bertsekas1989parallel}), for any fixed $x \in M$, we have
\begin{align*}
& \ \ \ \ \abs{\mathcal{T}\left( p^{(1)} \right) (x) - \mathcal{T}\left( p^{(2)}\right) (x) }\\
& = \abs{P_{\mathcal{E}} (p^{(1)} (x)) - P_{\mathcal{E}} (p^{(2)}(x))} \leq \abs{p^{(1)}(x) - p^{(2)}(x)}.
\end{align*}

Hence, based on the definition of the weighted $L^1$ norm, we can directly obtain
\begin{align*}
& \ \ \ \int\limits_{x \in M} \abs{\mathcal{T}\left( p^{(1)} \right) (x) - \mathcal{T}\left( p^{(2)} \right) (x)} w(x) d \lambda(x)\\
& \leq \int\limits_{x \in M} \abs{p^{(1)}(x) - p^{(2)}(x)} w(x) d \lambda(x),
\end{align*}
which is equivalent to~(\ref{ineq:nonexpansive}).
\end{proof}

Now we establish the existence and uniqueness of mean-field equilibrium in the case of CDMA communication protocol.
\begin{theorem}
Assume $\alpha < 1$. Then for CDMA, there exists a unique mean-field equilibrium~$p^* \in \mathcal{A} \subset L^1(M, \mathbb{R},\nu)$ with single user detection and utility function~(\ref{eq:utility_CDMA}). Moreover, starting from any initial strategy~$p_0 \in \mathcal{A}$, the unique mean-field equilibrium~$p^*$ can be obtained through a recursive update based on the best response operator, i.e., $\lim\limits_{k \to \infty} p_k = p^* \in L^1(M, \mathbb{R},\nu)$, where $p_{k+1} (\theta_i) \in \mathcal{BR}(\theta_i, p_k)$ for any~$\theta_i \in M$ and $k \geq 0$.
\thlabel{THM:CDMA_contraction_mapping}
\end{theorem}

\begin{proof}
See Appendix~\ref{sec:proof_thm1}.
\end{proof}

\begin{remark}
Note that the assumption~$\alpha < 1$ can be interpreted as requiring the length of the signature sequence~$n_s$ to be greater than the number of users~$N$. In other words, there is some redundancy in the number of sub-carries (i.e., code chips) in the implementation of spread spectrum technique. This assumption is commonly satisfied in the practical implementation of CDMA protocol, as it is often required that $n_s \gg N$ to enhance robustness of orthogonality against time offsets among signal waveforms from different users or to accommodate possible incoming users.
\end{remark}

Moreover, we also obtain the result on the continuity and monotonicity of the equilibrium strategy profile under the CDMA protocol.

\begin{corollary}
The mean-field equilibrium strategy $p^*: M \to \mathcal{E}$ for CDMA is continuous and monotonically increasing with respect to the identifier $\theta_i \in M$.
\thlabel{coro:CDMA_strategy}
\end{corollary}

\begin{proof}
See Appendix~\ref{sec:proof_coro1}.
\end{proof}

\begin{remark}
Since our model uses an infinite number of users to approximate the behavior of a large but finite user population, it is natural to raise the question on how accurate this approximation is. We take CDMA for a finite population with~$N$ users as an instance. Then we define the mean-field interference as the average receiving power level of the user population over the shared Gaussian channel, i.e.,~$z := (1/N) \sum_{j=1}^N p(\theta_j) \theta_j$, which is identical for all channel users. On the other hand, each user~$\theta_i$ chooses its optimal action in response to the interference~$z$, i.e., $p(\theta_i) = a_i^* = \argmax_{a_i \in \mathcal{E}} u(a_i,p,\theta_i) = \frac{1}{\beta \ln2} - \frac{\alpha z + N_0}{\theta_i}$, where we assume the set of feasible power levels~$\mathcal{E}$ is sufficiently large such that no truncation is performed. Hence, by calculating the fixed point of the mean-field term~$z$, a closed-form expression of the equilibrium strategy can be obtained as~$p^*(\theta_i) := \frac{1}{\beta \ln2} - \frac{1}{\theta_i}[\frac{N_0}{1+\alpha} + \frac{\alpha}{(1+\alpha) \beta \ln2} (\frac{1}{N} \sum_{i=1}^N \norm{h_i}^2)]$ for~$i=1,2,\ldots,N$. As $N \to \infty$, based on the strong law of large number, we have $\frac{1}{N} \sum_{i=1}^N \norm{h_i}^2 \xrightarrow[]{\text{a.s.}} \mathbb{E}[\norm{h_i}^2]$. According to Theorem~2.5.7 in~\cite{durrett2010probability}, the speed of almost sure convergence is faster than~$N^{-0.5} (\log N)^{0.5+\epsilon}$ for any~$\epsilon>0$.
\end{remark}

\section{NOMA transmission power game}
\label{sec:NOMA}
\subsection{Multi-packet reception (MPR) in NOMA}
\subsubsection{Power-domain NOMA and SIC decoding}
For power-domain NOMA, we still adopt the spread-spectrum technique. For convenience and fairness of social welfare comparison between CDMA and NOMA, the same set of signature sequences $\{\bm{s}_j\}_{j=1}^N$ is allocated to the users, hence $0< \alpha \leq 1$. The main difference between CDMA and NOMA is reflected in the specific decoding algorithm adopted by the receiver.

Now we explicitly give the procedures of SIC decoding algorithm for power-domain NOMA. This requires a proper determination of the decoding order of signals from different sources. Typically, the order of decoding is affected by the realization of the channel gains $\norm{h_i}^2$, as indicated in Vaezi et al.~\cite{vaezi2018non} and Xia et al.~\cite{xia2018outage}. The outcome of decoding is based on the $\sinr$ and is determined by Shannon's theorem as given in~(\ref{ineq:Shannon}). The $\sinr$ level of each data source in the received signal is presented below.

Consider the decoding order of the signals as a vector~$\bm{v}:=~(v_1,v_2,\ldots,v_N)$, where each index~$v_j \in~\{1,2,\ldots,N\}$ is distinct. For conventional SIC schemes, the recursive decoding continues if and only if all the previous decoding procedures are successful. From the perspective of the signal source $v_i$ which is currently being decoded, the set of signal sources which has been successfully decoded and canceled from received waveform $y[k]$ is $I_d(v_i)=\{v_1, v_2, \ldots, v_{i-1} \}$. Then the $\sinr$ of each signal source upon the decoding procedure can be expressed as
\begin{align}
\sinrvi & = \frac{p_{v_i}}{\sum\limits_{j \in \{1,2,\ldots,N\} / \{v_i\} / I_d(v_i)} \rho_{j, v_i} a_j \norm{h_j}^2 + N_0} \nonumber\\
& = \frac{p_{v_i}}{\sum\limits_{j \in \{1,2,\ldots,N\} / \{v_i\} / I_d(v_i)} \frac{\alpha}{N} a_j \norm{h_j}^2 + N_0}.
\end{align}

If the signal from source $m_i$ is decoded successfully from the received waveform $y[k]$, i.e., $\gamma^{(m_i)}=1$, the SIC decoding will proceed to the one with the largest channel gain $\norm{h_{m_{i+1}}}^2$ among remaining signal sources; otherwise, SIC terminates and all the signals from the remaining sources in the received waveform are dropped.

In order to extract information more effectively from the received waveform $y[k]$, an improved version of SIC is proposed by Xia et al.~\cite{xia2018outage} for uplink transmission such that the base station attempts to decode the remaining users' information even when failures have happened. In other words, the signal that was not successfully decoded will be treated as interference in the subsequent decoding procedures. In this case, the successful decoding set from the perspective of user $v_i$ is expressed as
\begin{equation}
\tilde{I}_d (v_i) = \{v_j : 1 \leq j \leq i-1 \text{ and } \gamma^{(j)} = 1\}
\end{equation}
and the decoding procedure terminates after all data sources~$i \in~\{1,2,\ldots,N\}$ have been attempted for decoding. We base our analysis on this improved version of SIC~\cite{xia2018outage}.

The implementation issues of the SIC decoding algorithm is also of concern. Though the complexity introduced by SIC cannot be ignored~\cite{vaezi2018non}, the online implementation of SIC under the power-domain NOMA considered in this paper is feasible due to the following reasons. First, recent improvements in the computational capability have enabled the implementations of SIC at user equipments~\cite{vaezi2018non}, such as the network-assisted interference cancellation and suppression (NAICS) terminals investigated by Zhou et al.~\cite{zhou2014network}. Second, the non-cooperative channel access model corresponds to the scenario of a micro base station serving a moderate user population size, typically around a hundred users. For such a user population, it is feasible to adopt SIC decoding during online implementation.

In practice, the base station senses the channel gain of the communication link from each user in advance, and it is relatively static since the users lack mobility. Hence, the base station can determine the decoding order of each user in advance and broadcast the channel gains as well as the decoding order to all potential users during initialization. This fixed decoding order assumption for SIC has also been employed by Wei et al.~\cite{wei2018performance}.


\begin{remark}
For convenience of analysis, we have assumed that a fixed decoding order (i.e., the descending order of the channel gain) is employed at the base station. However, from the perspective of a practical implementation, as the channel users are non-cooperative decision makers, it is in general difficult to regulate different users to choose their actions based on the predetermined SIC decoding order broadcasted by the base station. The requirement that each channel user will automatically adopt the decoding order broadcasted by the base station can be interpreted as a communication protocol programmed into the user equipments.

	On the other hand, NOMA with a fixed SIC decoding order can be considered as a benchmark for possible performance achievable by NOMA. Specifically, if an optimal decoding order for NOMA exists, it will perform no worse than the benchmark. Hence, we can without loss of generality analyse NOMA with a fixed decoding order.
			
\end{remark}

\subsubsection{The utility functions for the aggregative game using SIC}

To determine the utility function of each uplink user, we need to obtain the expression of $\sinr$ of each received signal sources at the base station. Due to the recursive nature of SIC, it is necessary to determine the decoding order before giving the expressions of $\sinr$. In this paper, we restrict our consideration to the case where the SIC algorithm at the base station follows the descending order of the squared norm of the channel gain.

For simplicity of analysis, we assume that at each step of SIC, the interference caused by users decoded prior to this step is perfectly canceled regardless of their decoding outcomes, which is similar to the model in Chen et al.~\cite{chen2017optimization}. Thus the \rm SINR is approximated by
\begin{equation}
\overline{\sinri}= \frac{a_i}{\frac{\alpha}{N} \sum\limits_{j \neq i} a_j \norm{h_j}^2 \cdot \bm{1}_{\{\norm{h_j}^2 < \norm{h_i}^2 \}} + N_0}.
\label{eq:SINR_aggre_NOMA1}
\end{equation}

\begin{remark}
Though we have assumed that each user enjoys a perfect cancellation of interference from previously decoded users during SIC, this does not hold true in general as the receiving $\sinr$ of a portion of users can fail to meet condition~(\ref{ineq:Shannon}) for successful decoding. The model we consider here is an approximation, which provides a straightforward characterization of the features of NOMA without introducing much complexity in modelling. Fortunately, if there are more general models where the interference faced by each user during SIC can be characterized explicitly, the social welfare comparison will remain valid as long as equilibrium strategies exist.
\end{remark}

Similar to the case of CDMA, when the number of players~$N \to \infty$ under the NOMA protocol, the~$\sinr$ of the received signal from player~$\theta_i \in M$ is expressed as
\begin{align*}
\overline{\sinr} (p,\theta_i) & = \lim\limits_{N \to \infty} \frac{a_i}{\frac{\alpha}{N} \sum\limits_{j \neq i}  a_j \theta_j \bm{1}_{\{\theta_j < \theta_i \}} + N_0}\\
& \xrightarrow[\text{SLLN}]{\text{a.s.}} \frac{a_i}{\alpha \mathbb{E}[p(\theta^*) \theta^* \bm{1}_{\{\theta^* < \theta_i \}}] + N_0}, \ \ \theta_i \in M,
\end{align*}
where the expectation is taken with respect to the random variable~$\theta^*$ following the distribution~$f(x)$, and Kolmogorov's strong law of large number holds since by \thref{assumption:channel_gain},
\begin{equation*}
\mathbb{E}[p(\theta^*) \theta^* \bm{1}_{\{\theta^* < \theta_i\}}] \leq \mathbb{E} [p(\theta^*) \theta^*] \leq E_{\max} \overline{h}^2< \infty.
\end{equation*}

The utility function of any user $\theta_i \in M$ under the NOMA protocol is
\begin{align}
& \ \ u (a_i , p, \theta_i) =\log_2 \left( 1 + \theta_i \overline{\sinr} (p,\theta_i) \right) - \beta a_i \nonumber\\
&=\log_2 \left( 1 + \frac{\theta_i a_i}{\alpha \mathbb{E}[p(\theta^*) \theta^* \bm{1}_{\{\theta^* < \theta_i\}}] + N_0} \right) - \beta a_i.
\label{eq:utility_NOMA}
\end{align}

\begin{remark}
The determination of the interference $\alpha \mathbb{E}[p(\theta^*) \theta^* \bm{1}_{\{\theta^* < \theta_i\}}]$ faced by each user $\theta_i \in M$ is equivalent to the identification of the decoding order of that specific user, as the interference is the only term in the utility function impacted by the decoding order. According to the strong law of large number, the interference from an infinite number of users will converge almost surely to the term $\alpha \mathbb{E}[p(\theta^*) \theta^* \bm{1}_{\{\theta^* < \theta_i\}}]$. Hence, each user $\theta_i \in M$ can determine the percentage of users that complete decoding before its turn as long as the PDF $f(\theta_i), \ (\theta_i \in M)$ is known.
\end{remark}

\subsection{Analysis of the NOMA transmission power equilibrium}

Player $\theta_i$'s optimal action against opponents' aggregate actions is given by the best response operator~(\ref{def:BR}), i.e.,~$\forall \theta_i \in M$, there is
\begin{align}
a_i^* &\in \mathcal{BR}_{\text{ordered}} (\theta_i, p) \nonumber\\
&=\bigg\{ P_{\mathcal{E}} \left( \frac{1}{\beta \ln 2} - \frac{\alpha \mathbb{E}[p(\theta^*) \theta^* \bm{1}_{\{\theta^* < \theta_i\}}] + N_0}{\theta_i} \right) \bigg\}.
\label{eq:BR_NOMA2}
\end{align}

Again, we are going to adopt the Banach fixed point theorem to establish the existence and uniqueness of the mean-field equilibrium strategy $p^*$.

According to~\thref{def:ess_supremum_norm}, a vector space consisting of all strategies with bounded power allocation is defined as
\begin{equation}
L^{\infty}(M, \mathbb{R},\nu):=\{ p: (p:M \to \mathbb{R}) \ \& \ (\norm{p}_{\infty}^{\nu} < \infty) \}.
\end{equation}

As explained in \thref{remark:Lp}, the space~$\left(L^{\infty}(M, \mathbb{R},\nu), \norm{\cdot}_{\infty}^{\nu} \right)$ is a Banach space. The set of feasible strategies $\mathcal{A}$ is a subset of $L^{\infty}(M, \mathbb{R},\nu)$.

Now we give the main result concerning the equilibrium strategy for NOMA, which establishes the existence and uniqueness of equilibrium strategy.
\begin{theorem}
Assume $\alpha < 1$. Then for the game~$\mathcal{G}$ which adopts utility function~(\ref{eq:utility_NOMA}) and NOMA with SIC decoding by descending order of the channel gains~$\norm{h_i}^2$, there exists a unique mean-field equilibrium $p^*_{\text{ordered}} \in \mathcal{A} \subset L^{\infty}(M, \mathbb{R},\nu)$, and the utility function of each player is given by~(\ref{eq:utility_NOMA}). Moreover, starting from any initial strategy~$p_0 \in \mathcal{A}$, the unique mean-field equilibrium~$p^*_{\text{ordered}}$ can be obtained through a recursive update based on the best response operator, i.e., $\lim\limits_{k \to \infty} p_k = p^*_{\text{ordered}} \in L^{\infty}(M, \mathbb{R},\nu)$, where $p_{k+1} (\theta_i) \in \mathcal{BR}_{\text{ordered}} (\theta_i, p_k)$ for any~$\theta_i \in M$ and $k \geq 0$.
\thlabel{THM:NOMA2_contraction_mapping}
\end{theorem}

\begin{proof}
See Appendix~\ref{sec:proof_thm2}.
\end{proof}

\begin{corollary}
For the game $\mathcal{G}$ adopting NOMA, the unique equilibrium strategy $p^*_{\text{ordered}}: M \to \mathcal{E}$ characterized in~\thref{THM:NOMA2_contraction_mapping} is continuous with respect to $\theta_i \in M$.
\thlabel{coro:NOMA_strategy}
\end{corollary}

\begin{proof}
See Appendix~\ref{sec:proof_coro2}.
\end{proof}

Based on \thref{THM:CDMA_contraction_mapping} and \thref{THM:NOMA2_contraction_mapping}, it is seen that the strategy update with a best response operator is a contraction mapping under both CDMA and NOMA protocols. According to Banach fixed-point theorem~\cite{bertsekas1989parallel}, this has directly led to a distributed algorithm for strategy updates at each player such that the convergence to the unique equilibrium strategy is guaranteed. The distributed algorithm is given in Algorithm~\ref{algo:distributed_NE}.

\begin{algorithm}[h]
\SetAlgoLined
\textbf{Input:} Number of users~$N$\;
\ \ \ \ \ \ \ \ \ Number of iterations~$NUM_{Itr}$\;
\KwResult{The strategy profile~$p^*(\theta_i)$ (CDMA) or~$p^*_{\text{ordered}}(\theta_i)$ (NOMA), where~$i=1,2,\ldots,N$.}
 \textbf{Initialization}: Fix an arbitrarily chosen initial power allocation strategy~$p_0(\theta_i) \in \mathcal{A}$ for users~$i=1,2,\ldots,N$\;
 \For{$k=0$ to~$NUM_{Itr} - 1$}{
 \For{$i=1$ to~$N$}{
 $p_{k+1} (\theta_i) \in \mathcal{BR} (\theta_i, p_k)$ (CDMA) or
 $p_{k+1} (\theta_i) \in \mathcal{BR}_{\text{ordered}} (\theta_i, p_k)$ (NOMA)\;
 (Parallel updates for users~$i=1,2,\ldots,N$.)
 }}
 The strategy profile obtained is~$p^*(\theta_i) \in \mathcal{A}$ (CDMA) or~$p^*_{\text{ordered}} \in \mathcal{A}$ (NOMA), for~$i=1,2,\ldots,N$.
 \caption{Distributed Equilibrium-Seeking Algorithm}
 \label{algo:distributed_NE}
\end{algorithm}

	Next, we will compare between OMA and NOMA for a given non-cooperative user population. Since a subcarrier (e.g., frequency band, time slots, signature sequences, etc.) needs to be allocated to each user before we model the channel interference, we employ CDMA with single-user decoding as a benchmark of OMA protocols. Due to the averaging effect introduced by the spread spectrum techniques, the channel interference faced by users can be described by a mean-field term. The same spread spectrum techniques is adopted in the subcarrier allocation of the NOMA scheme for the fairness of comparison. In this paper, we take power-domain NOMA with SIC as the representative of general NOMA schemes. The results are expected to be valid for other types of subcarriers as well.

\section{Social welfare comparison between CDMA and NOMA}
\label{sec:performance_comparison}
Now that the game $\mathcal{G}$ under either CDMA or NOMA communication protocol admits a unique mean-field equilibrium, it is of interest to conduct social welfare comparison when a large number of players reach an equilibrium under these two protocols respectively. In this paper, we focus more on qualitative analysis than quantitative analysis.

In general, the social welfare comparison between mean-field equilibria of different games is not easily achievable since it will be difficult to characterize the changes in equilibrium strategies. In an aggregative game with a large number of players, we define the social welfare as the average utility achieved by all players.

In this paper, we aim at comparing the effectiveness of the NOMA communication protocol in 5G networks against the CDMA protocol. As introduced in the previous sections, the intrinsic difference between NOMA and CDMA is whether successive interference cancellation (SIC) is adopted. The social welfare comparison of the game equilibria achieved under two different communication protocols can be formulated as two optimization problems with a common objective function (i.e., the social welfare metric), several different constraints (i.e., different communication protocols) as well as a common constraint reflecting the definition of a mean-field equilibrium, which restricts the solution set of each optimization problem to be within the set of mean-field equilibria. The mathematical details are illustrated below.

First, we define the social welfare metric in terms of 
\begin{equation}
J(p,z) := \mathbb{E}[\tilde{u} (p(\theta_i),z,\theta_i)] = \int\limits_{\theta_i \in M} \tilde{u}(p(\theta_i),z,\theta_i) f(\theta_i) d \theta_i,
\label{eq:expected_utility}
\end{equation}
where~$p \in \mathcal{A}$ and~$z: M \to \mathbb{R}$ is~$\nu$-measurable, with the individual utility $\tilde{u}$ corresponding to each player $\theta_i$ defined as a function of the action $p(\theta_i)$ taken by player $\theta_i$ and the interference effects~$z$. The expressions of individual utilities~$\tilde{u}$, based on the formulation of the aggregative game, can be expressed as
\begin{equation}
\tilde{u} (a_i, z, \theta_i) := \log_2 \left( 1 + \theta_i \frac{a_i}{\alpha z(\theta_i) + N_0} \right) - \beta a_i.
\label{eq:individual_utilities_optimization}
\end{equation}

\begin{remark}
For finite users, the social welfare can be defined as the average utility achieved in the population~\cite{cardaliaguet2019efficiency}. When the number of users~$N$ approaches infinity, the social welfare converges to the expectation of individual utilities with respect to the distribution of the user identity~$\theta_i \in M$.
\end{remark}

The only differences between CDMA and NOMA is on the interference effects~$z: M \to \mathbb{R}$.

First, we consider the CDMA protocol. As the channel is shared in an orthogonal manner with spread spectrum techniques, the interference level is identical for different channel users, i.e.,
\begin{equation}
z(\theta_i) = \mathbb{E}[p(\theta_j) \theta_j], \ \forall \theta_i \in M.
\label{eq:interference_CDMA}
\end{equation}

Next, we consider the NOMA protocol. Since the SIC decoding algorithm is used by the base station, as introduced in Section~\ref{sec:NOMA}, the interferences faced by different channel users~$\theta_i \in M$ are non-identical. With slight abuse of notations, we arbitrarily fix a user identity~$\theta_i$ to be a constant rather than a random variable. Due to the assumption on perfect interference cancellation and fixed decoding order (i.e., descending order of the channel gain), the interference faced by each user under NOMA is
\begin{equation}
z(\theta_i) = \mathbb{E}[p(\theta_j) \theta_j \bm{1}_{ \{ \theta_j < \theta_i \}}], \ \forall \theta_i \in M,
\label{eq:interference_NOMA}
\end{equation}
where the indicator function~$\bm{1}_{\{\theta_j < \theta_i\}}$ comes from recursive cancellation of previously decoded signal.

To compare the equilibrium social welfare under CDMA and NOMA, it is difficult to calculate their equilibrium strategies in closed form and evaluate the corresponding equilibrium social welfare~(\ref{eq:expected_utility}). Hence, it is necessary for us to propose approaches for analyzing the trend of changes in the equilibrium performance as we switch the protocol from CDMA to NOMA.

Given an interference profile~$z(\theta_i) \ (\theta_i \in M)$, if we plan to predict the outcome of power control game among channel users and evaluate the corresponding performance, we need to restrict the strategies to satisfy the definition of mean-field equilibrium~(\thref{def:MFE}), i.e.,
\begin{equation}
\tilde{u} (p(\theta_i),z,\theta_i) \geq \tilde{u} (a_i,z,\theta_i), \ \forall a_i \in \mathcal{E}, \ \theta_i \in M.
\label{eq:BR_constraint}
\end{equation}

With the social welfare~$J(p,z)$ as the performance criterion of this game, our objective in this section is to theoretically compare~$J(p,z)$ achieved under the equilibrium strategies of CDMA and NOMA, respectively. We employ a perturbation-based approach on the social welfare~$J(p,z)$ as a functional with respect to~$(p,z)$ for characterizing the trend by which the equilibrium social welfare changes when we switch the protocol from CDMA to NOMA. Details are given in \thref{lemma:nec_optimal} and illustrated in Fig.~\ref{fig:perturbation}.

\begin{figure}[h!]
	\begin{center}
		\includegraphics[width=0.5\textwidth]{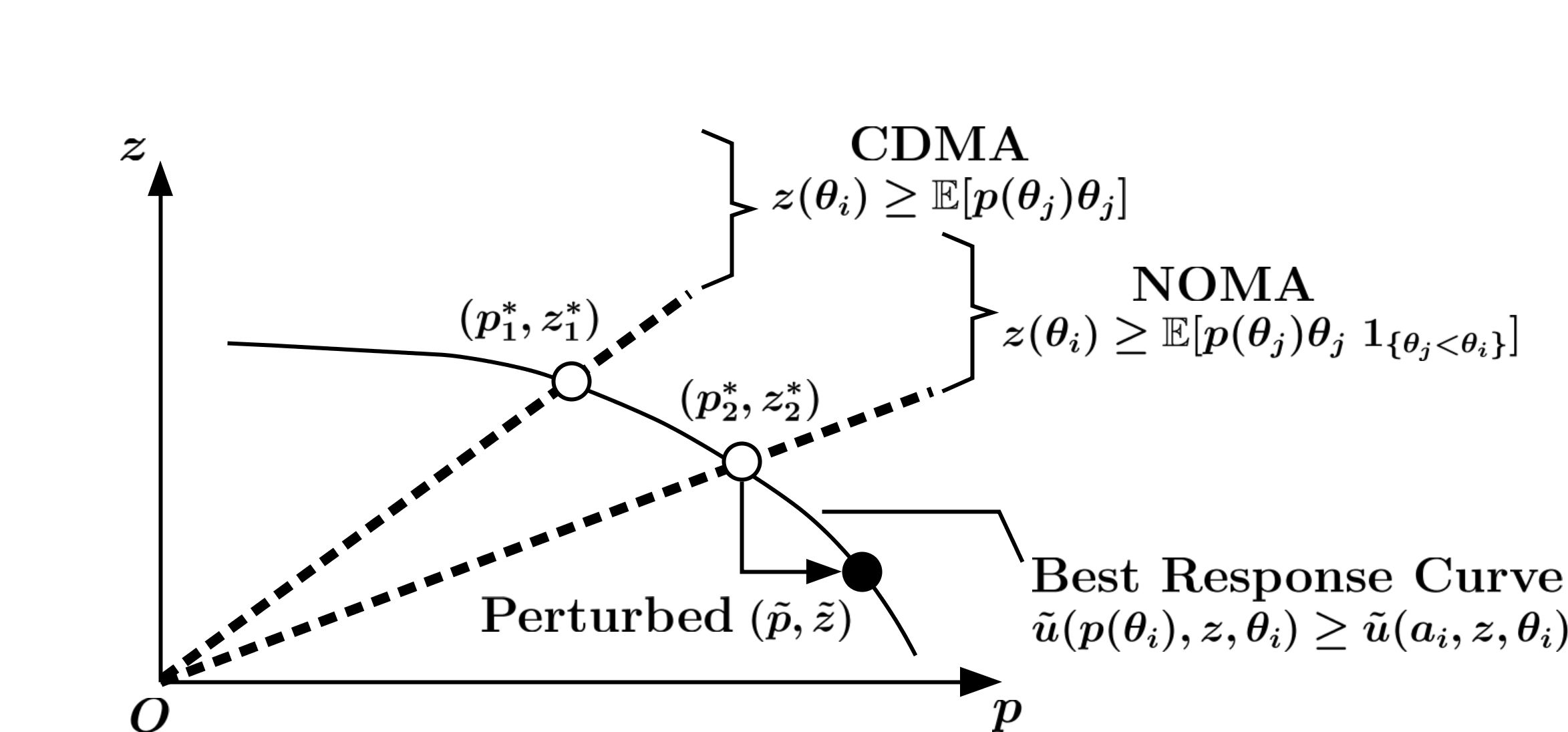}	
		\caption{Relaxed spaces of feasible function-valued variable pairs~$(p,z)$ and strategy perturbation.}
		\label{fig:perturbation}
	\end{center}
\end{figure}

Given the unique equilibrium power allocation strategies~$p^*$ and~$p^*_{\text{ordered}}$ in power control game~$\mathcal{G}$ under CDMA and NOMA, the corresponding interference profiles are given by~$z^*(\theta_i) = \mathbb{E}[p^*(\theta_j) \theta_j]$ and~$z_{\text{ordered}}^*(\theta_i) = \mathbb{E}[p_{\text{ordered}}^*(\theta_j) \theta_j \bm{1}_{\{\theta_j < \theta_i\}}]$ for any~$\theta_i \in M$, respectively.

\begin{theorem}[NOMA outperforms CDMA under equilibria]
NOMA can achieve a strictly better social welfare than CDMA at the corresponding equilibrium strategies, i.e.,~$J(p_{\text{ordered}}^*, z_{\text{ordered}}^*) > J(p^*, z^*)$.
\thlabel{lemma:nec_optimal}
\end{theorem}

\begin{proof}
To build up a bridge between the equilibrium performance of the power control game under CDMA and NOMA, we relax the conditions on interference profile~$z(\theta_i)$ such that only lower bounds are imposed, i.e.,~$z(\theta_i) \geq \mathbb{E}[p(\theta_j) \theta_j], \ \forall \theta_i \in M$ for CDMA and~$z(\theta_i) \geq \mathbb{E}[p(\theta_j) \theta_j \bm{1}_{\{\theta_j < \theta_i\}}], \ \forall \theta_i \in M$ for NOMA.

Under the relaxed interference profile~$z$, given the common best response condition~(\ref{eq:BR_constraint}) for both protocols, the space of feasible variable pairs~$(p,z)$ in CDMA is a subset of that in NOMA due to the fact that~$\mathbb{E}[p(\theta_j) \theta_j \bm{1}_{\{\theta_j < \theta_i\}}] \leq \mathbb{E}[p(\theta_j) \theta_j ], \ \forall \theta_i \in M$, as indicated in Fig.~\ref{fig:perturbation}.

Then we show that the maximum social welfare in the relaxed space of feasible variable pairs~$(p,z)$ under either CDMA or NOMA can only be obtained at the boundary of the relaxed conditions on interference~$z$. For convenience of presentation, given the relaxed variable space and the social welfare criterion, we construct two auxiliary optimization problems (\thref{problem:CDMA} and \thref{problem:NOMA}) to search for the best possible performance. Then, it is equivalent to showing that for both protocols, the maximum social welfare we seek can only be achieved on the boundary of the feasible sets.

\begin{problem}[CDMA - relaxed]
\begin{equation}
\begin{aligned}
& \max\limits_{p \in \mathcal{A}, \ z \in \mathcal{M}_{\mathbb{R}}(M)}
& & J(p,z)\\
& \text{subject to}
& & z(\theta_i) \geq \mathbb{E}[p(\theta_j) \theta_j], \\
&&& \tilde{u} (p(\theta_i),z,\theta_i) \geq \tilde{u} (a_i,z,\theta_i),\\
&&& \forall \ a_i \in \mathcal{E}, \ \theta_i \in M.
\end{aligned}
\label{eq:problem_CDMA}
\end{equation}
\thlabel{problem:CDMA}
\end{problem}

\begin{problem}[NOMA - relaxed]
\begin{equation}
\begin{aligned}
& \max\limits_{p \in \mathcal{A}, \ z \in \mathcal{M}_{\mathbb{R}}(M)}
& & J(p,z)\\
& \text{subject to}
& & z(\theta_i) \geq \mathbb{E}[p(\theta_j) \theta_j \bm{1}_{ \{ \theta_j < \theta_i \}}], \\
&&& \tilde{u} (p(\theta_i),z,\theta_i) \geq \tilde{u} (a_i,z,\theta_i),\\
&&& \forall \ a_i \in \mathcal{E}, \ \theta_i \in M.
\end{aligned}
\label{eq:problem_NOMA}
\end{equation}
\thlabel{problem:NOMA}
\end{problem}

Below, we establish the necessary optimality condition for NOMA, while a similar approach applies to CDMA.

To begin with, we focus on the constraint $z(\theta_i) \geq \mathbb{E}[p(\theta_j) \theta_j \bm{1}_{\{ \theta_j < \theta_i \}}]$. We aim at showing that any pair of optimal solution~$(p^*,z^*)$ to \thref{problem:NOMA} satisfies $z^*(\theta_i) = \mathbb{E} [p^*(\theta_j) \theta_j \bm{1}_{\{ \theta_j < \theta_i \}}]$ a.e. in~$M$.

From the feasible set of \thref{problem:NOMA}, we pick up a pair of decision variables~$(p,z)$ such that~$\tilde{u} (p(\theta_i),z,\theta_i) \geq \tilde{u} (a_i,z,\theta_i), \ \forall \theta_i \in M$ and there exists a bounded set~$\overline{M}_2 \subset M$ satisfying $P(\overline{M}_2) > 0$ and for any $\theta_i \in \overline{M}_2$, there is $z(\theta_i) > \mathbb{E} [p(\theta_j) \theta_j \bm{1}_{\{ \theta_j < \theta_i \}}]$. Assume~$(p,z)$ is an optimal solution to \thref{problem:NOMA}.

Define a measurable function~$\epsilon: \overline{M}_2 \to \mathbb{R}$ such that $\epsilon(\theta_i):=z(\theta_i) - \mathbb{E}[p(\theta_j) \theta_j \bm{1}_{\{\theta_j < \theta_i\}}] > 0$. Since the space~$M \subset \mathbb{R}$ is a metric space, according to Lemma~4.1 (Lusin's theorem) in Chapter~II of~\cite{parthasarathy2005probability}, for~$\epsilon_2:=\frac{1}{2} P(\overline{M}_2)>0$, there exists a closed set~$M_2 \subset \overline{M}_2$ such that~$\nu(\overline{M}_2 \setminus M_2) \leq \epsilon_2$ and the restriction of the measurable function~$\epsilon$ on the set~$M_2$, which is denoted as~$\epsilon_{M_2}: M_2 \to \mathbb{R}$, is continuous. Since~$\overline{M}_2$ is bounded, the closed set~$M_2 \subset \overline{M}_2$ is compact. Hence, based on Weierstrass extreme value theorem~\cite{luenberger1997optimization}, there exists a~$\theta' \in M_2$ such that~$\inf\limits_{\theta_i \in M_2} \epsilon_{M_2}(\theta_i) = \epsilon_{M_2}(\theta') = \epsilon(\theta') >0$.

Then, we construct a new variable~$\tilde{z}$ such that
\begin{align*}
\tilde{z} (\theta_i) :=
\begin{cases}
z(\theta_i) - \frac{1}{K} \epsilon (\theta_i), & \theta_i \in M_2; \cr
z(\theta_i), & \text{Otherwise},
\end{cases}
\end{align*}
where~$K>1$ is a scaling factor.

In order for the constructed variable~$\tilde{z}$ to satisfy the constraint $\tilde{u} (p(\theta_i),z,\theta_i) \geq \tilde{u} (a_i,z,\theta_i), \ \forall \theta_i \in M$, we obtain an updated version of the optimal power control variable~$\tilde{p}$ in response to the change in the interference term from~$z$ to~$\tilde{z}$. Since the individual utility function~$\tilde{u}(a_i,z,\theta_i)$ in the optimization problem is strictly concave with respect to the variable~$a_i$, it has a unique maximizer in terms of~$a_i$ when other variables are fixed. Then, the updated version of the optimal power control strategy $\tilde{p}(\theta_i)$ is expressed as
\begin{align}
\label{eq:argmax}
\tilde{p}(\theta_i) &:= \argmax\limits_{a_i \in \mathcal{E}} \ \tilde{u}(a_i,\tilde{z},\theta_i) \\
&= P_{\mathcal{E}} \left( \frac{1}{\beta \ln 2} - \frac{\alpha \tilde{z}(\theta_i) + N_0}{\theta_i} \right) \nonumber\\
&= \begin{cases}
P_{\mathcal{E}} \left( \frac{1}{\beta \ln 2} - \frac{\alpha [z(\theta_i) - \frac{1}{K} \epsilon(\theta_i)] + N_0}{\theta_i} \right), & \theta_i \in M_2; \cr
p(\theta_i), & \text{Otherwise}.
\end{cases} \nonumber
\end{align}

It remains to verify the existence of a scaling factor~$K>1$ such that the pair~$(\tilde{p},\tilde{z})$ satisfies the constraint $z(\theta_i) \geq \mathbb{E}[p(\theta_j) \theta_j \bm{1}_{\{\theta_j < \theta_i\}}]$ for any~$\theta_i \in M$. By definition of~$(\tilde{p},\tilde{z})$, it suffices to show that $\tilde{z}(\theta_i) \geq \mathbb{E}[\tilde{p}(\theta_j) \theta_j \bm{1}_{\{\theta_j < \theta_i\}}]$ for any~$\theta_i \in M_2$.

By the derivations in~(\ref{eq:argmax}), for any~$\theta_i \in M$,
\begin{align*}
& \ \ \ \abs{p(\theta_i) - \tilde{p}(\theta_i)}\\
&= \abs{\argmax\limits_{a_i \in \mathcal{E}} \ \tilde{u}(a_i,z,\theta_i) - \argmax\limits_{a_i \in \mathcal{E}} \ \tilde{u}(a_i,\tilde{z},\theta_i)}\\
& \leq \abs{\frac{\alpha [z(\theta_i) - \tilde{z}(\theta_i)]}{\theta_i}}
\leq \frac{\alpha \epsilon(\theta_i)}{K \theta_i}.
\end{align*}

Hence for any~$\theta_i \in M$,
\begin{align*}
& \ \ \ \abs{\mathbb{E}[(p(\theta_j) - \tilde{p}(\theta_j)) \theta_j \bm{1}_{\{\theta_j < \theta_i\}}]}\\
&\leq \mathbb{E}[\abs{p(\theta_j) - \tilde{p}(\theta_j)} \theta_j]
= \frac{\alpha}{K} \mathbb{E}[\epsilon(\theta_j)].
\end{align*}

As~$\epsilon(\theta_i) > 0$ for any~$\theta_i \in M_2$ and $0 \leq \mathbb{E}[\epsilon(\theta_j)] < \infty$ is a constant, there exists a sufficiently large~$K$ such that
\begin{align*}
& \ \ \ \tilde{z}(\theta_i) - \mathbb{E}[\tilde{p}(\theta_j)\theta_j \bm{1}_{\{\theta_j < \theta_i\}}]\\
&=z(\theta_i) - \frac{1}{K}\epsilon(\theta_i) - \mathbb{E}[p(\theta_j)\theta_j \bm{1}_{\{\theta_j < \theta_i\}}] +\\
& \ \ \ \mathbb{E}[(p(\theta_j) - \tilde{p}(\theta_j)) \theta_j \bm{1}_{\{\theta_j < \theta_i\}}]\\
&\geq \frac{K-1}{K} \epsilon(\theta_i) - \abs{\mathbb{E}[(p(\theta_j) - \tilde{p}(\theta_j)) \theta_j \bm{1}_{\{\theta_j < \theta_i\}}]}\\
&\geq \frac{K-1}{K} \epsilon_{M_2}(\theta') - \frac{\alpha}{K} \mathbb{E}[\epsilon(\theta_j)] > 0, \ \forall \theta_i \in M_2,
\end{align*}
where the last inequality holds due to~$\epsilon(\theta_i)=\epsilon_{M_2}(\theta_i) \geq \inf\limits_{\theta_i \in M_2} \epsilon_{M_2}(\theta_i) = \epsilon_{M_2}(\theta') $ for any~$\theta_i \in M_2$ and by a fixed choice of~$K > \frac{\alpha \mathbb{E}[\epsilon(\theta_j)]}{\epsilon(\theta')} + 1 \geq 1$. Thus, the feasibility of the constructed variable pair~$(\tilde{p},\tilde{z})$ is successfully shown.

Now, since $\tilde{z}(\theta_i) < z(\theta_i)$ by definition, we can obtain the following inequalities
\begin{align*}
\tilde{u} (p(\theta_i),z,\theta_i) &= \log_2 \left( 1 + \theta_i \frac{p(\theta_i)}{\alpha z(\theta_i) + N_0} \right) - \beta p(\theta_i)\\
&< \log_2 \left( 1 + \theta_i \frac{p(\theta_i)}{\alpha \tilde{z}(\theta_i) + N_0} \right) - \beta p(\theta_i)\\
&\leq \log_2 \left( 1 + \theta_i \frac{\tilde{p}(\theta_i)}{\alpha \tilde{z}(\theta_i) + N_0} \right) - \beta \tilde{p}(\theta_i)\\
&= \tilde{u} (\tilde{p}(\theta_i),\tilde{z},\theta_i), \ \forall \theta_i \in M_2,
\end{align*}
where the last inequality is due to the fact that $a_i=\tilde{p}(\theta_i)$ is a maximizer of $\tilde{u}(a_i,\tilde{z},\theta_i)$.

Therefore, the social welfare under the new decision variable satisfies
\begin{align*}
& J(\tilde{p},\tilde{z})
= \int\limits_{\theta_i \in M} \tilde{u}(\tilde{p}(\theta_i),\tilde{z},\theta_i) f(\theta_i) d \theta_i
>\\
&\int\limits_{\theta_i \in M_2} \tilde{u}(p(\theta_i),z,\theta_i) f(\theta_i) d \theta_i + \int\limits_{\theta_i \in M \setminus M_2} \tilde{u}(\tilde{p}(\theta_i),\tilde{z},\theta_i) f(\theta_i) d \theta_i\\
&= \int\limits_{\theta_i \in M_2} \tilde{u}(p(\theta_i),z,\theta_i) f(\theta_i) d \theta_i +\\
& \int\limits_{\theta_i \in M \setminus M_2} \tilde{u}(p(\theta_i),z,\theta_i) f(\theta_i) d \theta_i
\ = \  J(p,z),
\end{align*}
which indicates that it is not possible for the original pair of decision variables $(p,z)$ to be optimal.

Therefore, $z(\theta_i) = \mathbb{E} [p(\theta_j) \theta_j \bm{1}_{\{ \theta_j < \theta_i \}}]$ holds almost everywhere in $M$ is a necessary optimality condition of \thref{problem:NOMA}. In light of the proof above, a necessary condition for optimality can be obtained for \thref{problem:CDMA} such that $z(\theta_i) = \mathbb{E} [p(\theta_j) \theta_j]$ holds almost everywhere in~$M$.

Denote the optimal solution to \thref{problem:CDMA} as~$(p_1^*,z_1^*)$ and the optimal solution to \thref{problem:NOMA} as~$(p_2^*,z_2^*)$. Take into account the inclusive relationship between the relaxed spaces of feasible variables~$(p,z)$ for CDMA and NOMA, we can obtain~$J(p_2^*,z_2^*) \geq J(p_1^*,z_1^*)$. Next, we show that the above inequality is strict.

Assume~$J(p_2^*,z_2^*) = J(p_1^*,z_1^*)$, since the optimal variable pair~$(p_1^*,z_1^*)$ in \thref{problem:CDMA} is also feasible for \thref{problem:NOMA}. Then the optimal social welfare for NOMA can also be achieved at~$(p_1^*,z_1^*)$. However, since there exists a non-zero measure set~$M_+ \subset M$ such that~$z_1^*(\theta_i) = \mathbb{E} [p_1^*(\theta_j) \theta_j] \neq \mathbb{E} [p_1^*(\theta_j) \theta_j \bm{1}_{\{ \theta_j < \theta_i \}}], \ \forall \theta_i \in M_+$, this contradicts with the necessary optimality condition for \thref{problem:NOMA}. Thus, we have~$J(p_2^*,z_2^*) > J(p_1^*,z_1^*)$.

Next, we show that the social welfare achieved by the variable pairs~$(p_1^*,z_1^*), \ (p_2^*,z_2^*)$ equals that achieved under the equilibrium of CDMA, i.e.,~$(p^*,z^*)$, and NOMA, i.e.,~$(p_{\text{ordered}}^*,z_{\text{ordered}}^*)$, respectively. We show it for NOMA, as similar arguments can be followed for CDMA.

For the variable pair~$(p_2^*,z_2^*)$ achieving the optimal value of \thref{problem:NOMA}, there exists a subset $M_0 \subset M$ with $P(M_0) = 0$, where $z_2^*(\theta_i) = \mathbb{E}[p_2^*(\theta_j) \theta_j \bm{1}_{\{\theta_j < \theta_i\}}]$ holds for any~$\theta_i \in M \setminus M_0$. Hence, the strict inequality $z_2^*(\theta_i) > \mathbb{E}[p_2^*(\theta_j) \theta_j \bm{1}_{\{\theta_j < \theta_i\}}]$ can only be satisfied at some points~$\theta_i$ in~$M_0$.

Based on this variable pair, we construct an auxiliary variable pair $(\tilde{p},\tilde{z})$ as follows.
\begin{equation}
\tilde{z}(\theta_i) :=
\begin{cases}
z_2^* (\theta_i), & \theta_i \in M \setminus M_0;\cr
\mathbb{E}[p_2^*(\theta_j) \theta_j \bm{1}_{\{\theta_j < \theta_i\}}], & \text{Otherwise},
\end{cases}
\end{equation}
and
\begin{equation}
\tilde{p}(\theta_i) =
\begin{cases}
p_2^*(\theta_i), & \theta_i \in M \setminus M_0; \cr
P_{\mathcal{E}} \left( \frac{1}{\beta \ln 2} - \frac{\alpha \tilde{z}(\theta_i) + N_0}{\theta_i} \right), & \text{Otherwise}.
\end{cases}
\end{equation}

As a result, we can obtain that
\begin{equation}
\begin{cases}
\tilde{z}(\theta_i) = \mathbb{E} [p_2^*(\theta_j) \theta_j \bm{1}_{\{\theta_j < \theta_i\}}], & \forall \theta_i \in M; \cr
\tilde{p}(\theta_i) = P_{\mathcal{E}} \left( \frac{1}{\beta \ln 2} - \frac{\alpha \tilde{z}(\theta_i) + N_0}{\theta_i} \right), & \forall \theta_i \in M.
\end{cases}
\end{equation}

In order to verify that the pair of auxiliary variables $(\tilde{p},\tilde{z})$ satisfies the necessary optimality condition for \thref{problem:NOMA}, it suffices to show $\mathbb{E}[\tilde{p}(\theta_j) \theta_j \bm{1}_{\{\theta_j < \theta_i\}}]=\mathbb{E}[p_2^*(\theta_j) \theta_j \bm{1}_{\{\theta_j < \theta_i\}}]$.

We derive that
\begin{align*}
\mathbb{E}[\tilde{p}(\theta_j) \theta_j \bm{1}_{\{\theta_j < \theta_i\}}] &= \int\limits_{\theta_j \in M \cap (0,\theta_i]} \tilde{p}(\theta_j) \theta_j f(\theta_j) d \theta_j\\
&= 0 + \int\limits_{\theta_j \in (M/M_0) \cap (0,\theta_i]} p_2^*(\theta_j) \theta_j f(\theta_j) d\theta_j\\
&= \mathbb{E} [p_2^*(\theta_j) \theta_j \bm{1}_{\{\theta_j < \theta_i \}}],
\end{align*}
where the second equality holds due to the fact that the set $M_0$ has Lebesgue measure zero, i.e., $P(M_0)=0$.

As a result, the pair of auxiliary variables $(\tilde{p},\tilde{z})$ satisfies
\begin{equation}
\begin{cases}
\tilde{u}(\tilde{p}(\theta_i),\tilde{z},\theta_i) \geq \tilde{u}(a_i, \tilde{z}, \theta_i), & \forall a_i \in \mathcal{E}, \ \theta_i \in M; \cr
\tilde{z}(\theta_i) = \mathbb{E} [\tilde{p}(\theta_j) \theta_j \bm{1}_{\{\theta_j < \theta_i\}}],  & \forall \theta_i \in M.
\end{cases}
\end{equation}

These conditions correspond to the definition of mean-field equilibrium of the game~$\mathcal{G}$ when NOMA is adopted, and the mean-field equilibrium strategy for NOMA uniquely exists according to \thref{THM:NOMA2_contraction_mapping}. Hence, we can obtain that $(\tilde{p},\tilde{z}) = (p_{\text{ordered}}^*,z_{\text{ordered}}^*)$.

Now, since~$P(M_0)=0$, we analyse the value of the objective function achieved, as follows.
\begin{align*}
J(p_2^*,z_2^*) &= \int\limits_{\theta_i \in M} \tilde{u}(p(\theta_i),z,\theta_i) f(\theta_i) d \theta_i\\
&= 0 +\int\limits_{\theta_i \in M \setminus M_0} \tilde{u}(p(\theta_i),z,\theta_i) f(\theta_i) d \theta_i\\
&=0 + \int\limits_{\theta_i \in M \setminus M_0} \tilde{u}(\tilde{p}(\theta_i),\tilde{z},\theta_i) f(\theta_i) d \theta_i\\
&=J(\tilde{p},\tilde{z})=J(p_{\text{ordered}}^*,z_{\text{ordered}}^*).
\end{align*}

Therefore, the optimal solution~$(p_2^*,z_2^*)$ to \thref{problem:NOMA} achieves exactly the same value of the objective function~$J(p,z)$ as the mean-field equilibrium variable pair~$(p^*_{\text{ordered}},z^*_{\text{ordered}})$. With similar derivations for CDMA, we conclude that~$J(p_{\text{ordered}}^*,z_{\text{ordered}}^*) > J(p^*,z^*)$. This completes the proof.
\end{proof}

\section{Individual behaviors at the equilibrium}
\label{sec:individual_properties}
The collective behaviors comparison among the population of users has been conducted in the previous section. It is of interest to characterize individual behaviors at the equilibrium.

According to~\thref{THM:CDMA_contraction_mapping} and~\thref{THM:NOMA2_contraction_mapping}, denote the unique equilibrium strategy for CDMA as $p^* \in \mathcal{A}$ and the equilibrium strategy for NOMA as $p_{\text{ordered}}^* \in \mathcal{A}$. For an unbounded set $M$ of user identities, we obtain an additional property such that for users with sufficiently large uplink channel gains, their equilibrium transmission power under CDMA and NOMA can be arbitrarily close.

\begin{proposition}[Convergence behavior for high-gain users]
Assume the player set $M$ is unbounded from above, i.e., $\forall L>0, \ \exists  \theta_i \in M \ s.t. \ \theta_i > L$. Then
\begin{equation}
\lim\limits_{\theta_i \to \infty} \abs{p^*(\theta_i) - p^*_{\text{ordered}}(\theta_i)} = 0.
\end{equation}
\thlabel{prop:convergence}
\end{proposition}

\begin{proof}
By definition of the equilibrium strategies~$p^*$ and~$p^*_{\text{ordered}}$, we obtain the following inequality based on the best response operator~(\ref{def:BR}). For any~$\theta_i \in M$,
\begin{align*}
& \ \ \ \abs{p^*(\theta_i) - p^*_{\text{ordered}}(\theta_i)}\\
&= \frac{\alpha}{\theta_i} \abs{\mathbb{E}[p^*(\theta_j) \theta_j] - \mathbb{E}[p^*_{\text{ordered}}(\theta_j) \theta_j \bm{1}_{\{\theta_j < \theta_i\}}]}\\
&\leq \frac{\alpha}{\theta_i} \{ \abs{\mathbb{E}[p^*(\theta_j) \theta_j]} + \abs{\mathbb{E}[p^*_{\text{ordered}}(\theta_j) \theta_j]} \} \leq \frac{2\alpha}{\theta_i} E_{\max} \mathbb{E}[\norm{h_i}^2].
\end{align*}

For any~$\epsilon > 0$, we choose~$\theta_{\epsilon} > \frac{2 \alpha E_{\max} \mathbb{E}[\norm{h_i}^2]}{\epsilon}$. Then we obtain that~$\abs{p^*(\theta_i) - p^*_{\text{ordered}}(\theta_i)}  <\epsilon$ for any~$\theta > \theta_{\epsilon}$, which completes the proof.
\end{proof}

More importantly, it can be shown that the curve of equilibrium power allocation for different users under CDMA and NOMA will have a crossing. We can further prove that the curve of equilibrium data rate achieved under CDMA and NOMA crosses as well. In other words, pointwise improvement in the equilibrium data rate for different types of users is not achievable through adopting NOMA instead of CDMA. Intuitively, there is ``no free lunch" in employing NOMA to improve for all individual users.

\begin{proposition}[Infeasibility of pointwise improvement]
Assume~$E_{\min} = 0$. Then the curve of equilibrium power strategy~$p^*$ for CDMA crosses~$p^*_{\text{ordered}}$ for NOMA, i.e., $\exists \ \theta_{\text{cross}} \in M$ s.t. $p^*(\theta_{\text{cross}}) = p^*_{\text{ordered}} (\theta_{\text{cross}})$ and $\exists \ \theta_- < \theta_{\text{cross}} < \theta_+$ with $[p^*(\theta_-) - p^*_{\text{ordered}}(\theta_-)] \cdot [p^*(\theta_+) - p^*_{\text{ordered}}(\theta_+)] < 0$. Consequently, it is infeasible to achieve pointwise improvement in the equilibrium data rate achieved by NOMA in comparison with CDMA.

\thlabel{prop:cross}
\end{proposition}

\begin{proof}
First, we show that the curve of~$p^*$ crosses~$p^*_{\text{ordered}}$, which we will prove by contradiction.

Assume that the curves of~$p^*$ and~$p^*_{\text{ordered}}$ never cross each other. Define the cutoff thresholds $\theta_{th}^{\rm CDMA}, \theta_{th}^{\rm NOMA} \in M$ such that $p^*(\theta_i) = 0$ for any $\theta_i \leq \theta_{th}^{\rm CDMA}$ and $p^*_{\text{ordered}} (\theta_i) = 0$ for any $\theta_i \leq \theta_{th}^{\rm NOMA}$.

Since an equilibrium strategy is a best response~(\ref{eq:BR_CDMA}), (\ref{eq:BR_NOMA2}) to itself, if we set $p^*(\theta_{th}^{\rm CDMA}) = 0$ and $p^*_{\text{ordered}}(\theta_{th}^{\rm NOMA}) = 0$, we obtain that the cutoff thresholds satisfy~$0<\theta_{th}^{\rm NOMA} < \theta_{th}^{\rm CDMA}$. Thus, the nonexistence of crossing can be expressed as~$p^*(\theta_i) \leq p^*_{\text{ordered}}(\theta_i)$ for any~$\theta_i \in M$. Again, according to~(\ref{eq:BR_CDMA}), (\ref{eq:BR_NOMA2}),
\begin{equation}
p^*_{\text{ordered}}(\theta_i) = P_{\mathcal{E}} \left( \frac{1}{\beta \ln 2} - \frac{\alpha \mathbb{E}[p^*_{\text{ordered}}(\theta_j) \theta_j \bm{1}_{\{\theta_j < \theta_i\}}] + N_0}{\theta_i} \right),
\label{eq:cond_equilibrium_NOMA}
\end{equation}
it can be obtained that~$p^*(\theta_i) \leq p^*_{\text{ordered}}(\theta_i)$ implies~$\mathbb{E}[p^*(\theta_j) \theta_j] \geq \mathbb{E}[p^*_{\text{ordered}}(\theta_j) \theta_j \bm{1}_{\{\theta_j < \theta_i\}}]$, which holds for any~$\theta_i \in M$. By taking a sufficiently large~$\theta_i \in M$, we conclude that~$\mathbb{E}[p^*(\theta_j) \theta_j] \geq \mathbb{E}[p^*_{\text{ordered}}(\theta_j) \theta_j]$.

On the other hand, since~$p^*(\theta_i) \leq p^*_{\text{ordered}}(\theta_i)$, we have~$\mathbb{E}[p^*(\theta_j) \theta_j] \leq \mathbb{E}[p^*_{\text{ordered}}(\theta_j) \theta_j]$. Thus we have~$\mathbb{E}[p^*(\theta_j) \theta_j] =\mathbb{E}[p^*_{\text{ordered}}(\theta_j) \theta_j]$ as the inequality holds for both directions.

Given~$p^*(\theta_i) \leq p^*_{\text{ordered}}(\theta_i)$, we obtain that
\begin{align*}
& \ \ \  \mathbb{E}[\abs{p^*_{\text{ordered}}(\theta_j) - p^*(\theta_j)} \theta_j]\\
&=\mathbb{E}[p^*_{\text{ordered}}(\theta_j) \theta_j] - \mathbb{E}[p^*(\theta_j) \theta_j]=0.
\end{align*}

The expression above can be equivalently interpreted as
\begin{align*}
0 &= \mathbb{E}[\abs{p^*_{\text{ordered}}(\theta_j) - p^*(\theta_j)} \theta_j]\\
&= \int\limits_{\theta_j \in M} \abs{p^*_{\text{ordered}}(\theta_j) - p^*(\theta_j)} d\nu(\theta_j),
\end{align*}
i.e.,~$p^* = p^*_{\text{ordered}}$, $\nu$-a.e, where~$\nu$ is defined in~(\ref{eq:new_measure}).

The cutoff thresholds~$0<\theta_{th}^{\rm NOMA} < \theta_{th}^{\rm CDMA}$ are strictly different for~$p^*$ and~$p^*_{\text{ordered}}$, as shown above. Since both functions~$p^*:M \to \mathcal{E}$ and~$p^*_{\text{ordered}}:M \to \mathcal{E}$ are shown to be continuous with respect to~$\theta_i \in M$ in \thref{coro:CDMA_strategy} and \thref{coro:NOMA_strategy}, it is impossible to have~$p^* = p^*_{\text{ordered}}$, $\nu$-a.e.

Therefore, contradiction emerges, which verifies that the power allocation strategies~$p^*$ crosses~$p^*_{\text{ordered}}$. Denote the type variable at which~$p^*$ crosses~$p^*_{\text{ordered}}$ as~$\theta_{\text{cross}} \in M$, this is equivalent to $\exists \ \theta_- < \theta_{\text{cross}} < \theta_+$ with $[p^*(\theta_-) - p^*_{\text{ordered}}(\theta_-)] \cdot [p^*(\theta_+) - p^*_{\text{ordered}}(\theta_+)] < 0$.

Secondly, we show that pointwise improvement in the curve of equilibrium data rate achieved for different users is impossible.

The crossing behaviors cannot happen at saturation region, i.e.,~$E_{\min}$ or~$E_{\max}$. According to the best response equation~(\ref{eq:BR_CDMA}) and~(\ref{eq:BR_NOMA2}) as well as the continuity of~$p^*(\theta_i)$ and~$p^*_{\text{ordered}}(\theta_i)$ with respect to~$\theta_i \in M$, there exists a sufficiently small~$\delta>0$ such that for any~$\abs{\theta - \theta_{\text{cross}}} < \delta$, we have
\begin{align*}
p^*(\theta) = \frac{1}{\beta \ln 2} - \frac{\alpha \mathbb{E}[p^*(\theta_j) \theta_j] + N_0}{\theta}
\end{align*}
and
\begin{align*}
p^*_{\text{ordered}}(\theta) = \frac{1}{\beta \ln 2} - \frac{\alpha \mathbb{E}[p^*_{\text{ordered}}(\theta_j) \theta_j \bm{1}_{\{\theta_j < \theta\}}] + N_0}{\theta}.
\end{align*}

We denote the equilibrium data rate achieved by CDMA and NOMA at user type~$\theta_i \in M$ as~$d_{\rm {CDMA}}^*(\theta_i) = \log_2 \left( 1 + \theta_i \frac{p^*(\theta_i)}{\alpha \mathbb{E}[p^*(\theta_j) \theta_j] + N_0} \right)$ and~$d_{\rm {NOMA}}^*(\theta_i) = \log_2 \left( 1 + \theta_i \frac{p^*(\theta_i)}{\alpha \mathbb{E}[p^*_{\text{ordered}}(\theta_j) \theta_j \bm{1}_{\{\theta_j < \theta_i\}}] + N_0} \right)$ respectively. Then, based on the conditions on~$\theta_{\text{cross}}$ above, it can be obtained that for any~$\abs{\theta - \theta_{\text{cross}}} < \delta$,
\begin{equation}
d_{\rm {CDMA}}^*(\theta) = \log_2 \left( \frac{1}{1-\beta \ln2 \ p^*(\theta)} \right)
\end{equation}
and
\begin{equation}
d_{\rm {NOMA}}^*(\theta) = \log_2 \left( \frac{1}{1-\beta \ln2 \ p^*_{\text{ordered}}(\theta)} \right).
\end{equation}

The function~$d_{\rm {CDMA}}^*(\theta)$ and~$d_{\rm {NOMA}}^*(\theta)$ have the same monotonicity properties with respect to~$\theta \in M$ as~$p^*(\theta)$ and~$p^*_{\text{ordered}}(\theta)$ respectively when~$\abs{\theta - \theta_{\text{cross}}} < \delta$. Therefore, the curve of equilibrium data achieved for different users under CDMA crosses that under NOMA, i.e., pointwise improvement in the equilibrium data rate is infeasible through adopting NOMA instead of CDMA.
\end{proof}

\section{Simulations}
\label{sec:simulations}
In this section, we numerically illustrate the results concerning the properties of the equilibrium strategy profile under both CDMA (with fierce competition) and NOMA (with regulating effects among the user population).


First, we introduce some parameters and setups adopted in the simulation. We suppose that the channel gain $h_i$ for each user follows Rayleigh fading. Specifically, for an arbitrary user, the probability density function for the squared magnitude of its channel gain~$\theta_i=\norm{h_i}^2$ is
\begin{equation}
f(\theta_i) =
\begin{cases}
\frac{1}{\sigma} \exp(-\frac{\theta_i}{\sigma^2}), & \theta_i \geq 0; \cr
0, & \text{Otherwise}.
\end{cases}
\end{equation}

For the simulation we take the parameter $\sigma = 5$, and the probability density function is shown in the figure below.

\begin{figure}[H]
	\begin{center}
		\includegraphics[width=0.45\textwidth]{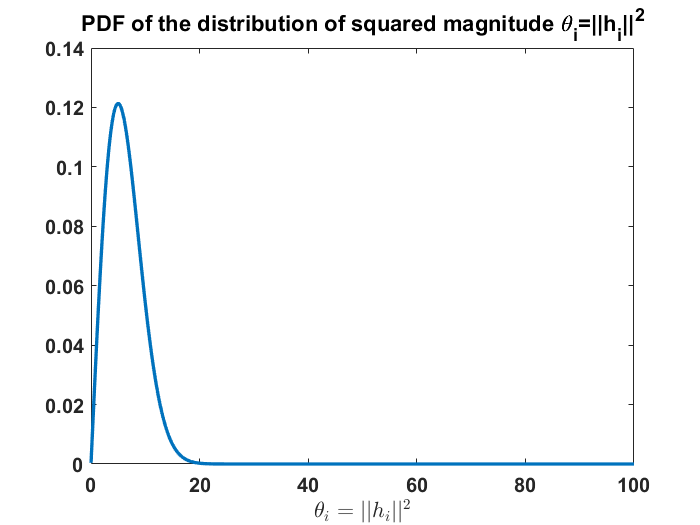}	
		\caption{The probability density function (PDF) of the squared magnitude of channel gain $\theta_i = \norm{h_i}^2$.}
		\label{fig:Rayleigh}
	\end{center}
\end{figure}

It is intractable to evaluate the behaviors of an infinite number of players for a numerical simulation, nor is it of interest in practice. Hence, the results we present below are generated with $N=1000$ players. The white noise process $w[k]$ in the additive white Gaussian noise (AWGN) channel features a power spectrum density $N_0 = 5$, and the spread spectrum parameter $\alpha = \frac{N}{n_s} = 0.25$ applies to both the case of CDMA and NOMA. In the aggregative game, we assume the set of feasible power levels is $\mathcal{E} = [0,150]$.

Now, we calculate the equilibrium power allocation strategy of the game~$\mathcal{G}$ as well as the corresponding data rates when CDMA and NOMA are adopted respectively, as shown in Fig.~\ref{fig:power_rate_compare}.

The equilibrium power allocation strategy is analysed first. It is noticed that the gap between the equilibrium strategy of CDMA and NOMA decreases as the value of the power penalty parameter~$\beta>0$ increases. An intuitive interpretation is that the parameter~$\beta$, which determines the cost of unit power consumption, will have a stronger regulating power when it takes a larger value because it results in a more conservative strategy for each user. Thus, through increasing the value of~$\beta$, fierce competitions in CDMA (i.e., high transmission power always results in a high data rate) can be relieved to a certain extent. Hence, the equilibrium power allocation under CDMA will gradually approach a natural fairness introduced through NOMA (i.e., signals from users with high channel gains or receiving gains benefit from their superiority of magnitude at the receiver, while others benefit from successive interference cancellation) as~$\beta$ increases. Besides, it is noticeable that the power consumption is significantly reduced with increase in the value of~$\beta$, for which an intuitive interpretation is the decrease of demand as unit price rises. The design problem of pricing in resource allocation has been investigated in~\cite{saraydar2002efficient,alpcan2002cdma,huang2006auction}.

\begin{figure*}[h!]
\begin{center}
\begin{subfigure}{}
\includegraphics[width=0.45\textwidth]{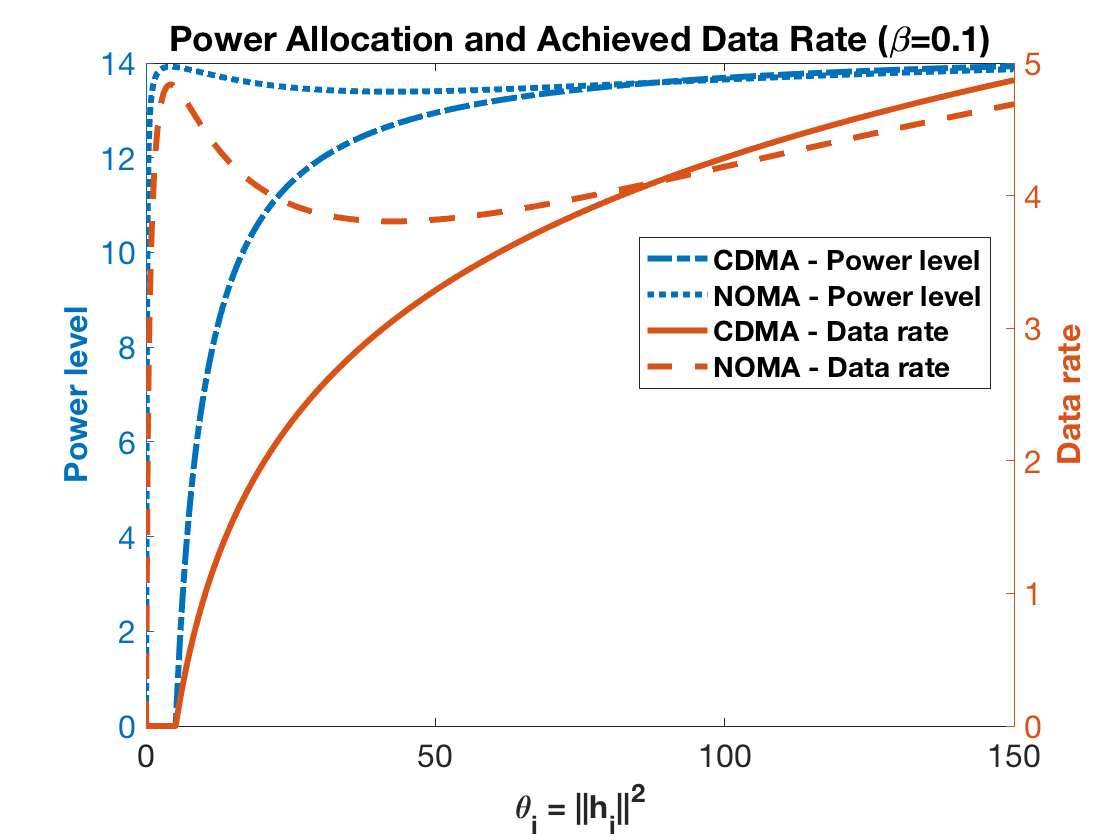}
\end{subfigure}
\begin{subfigure}{}
\includegraphics[width=0.45\textwidth]{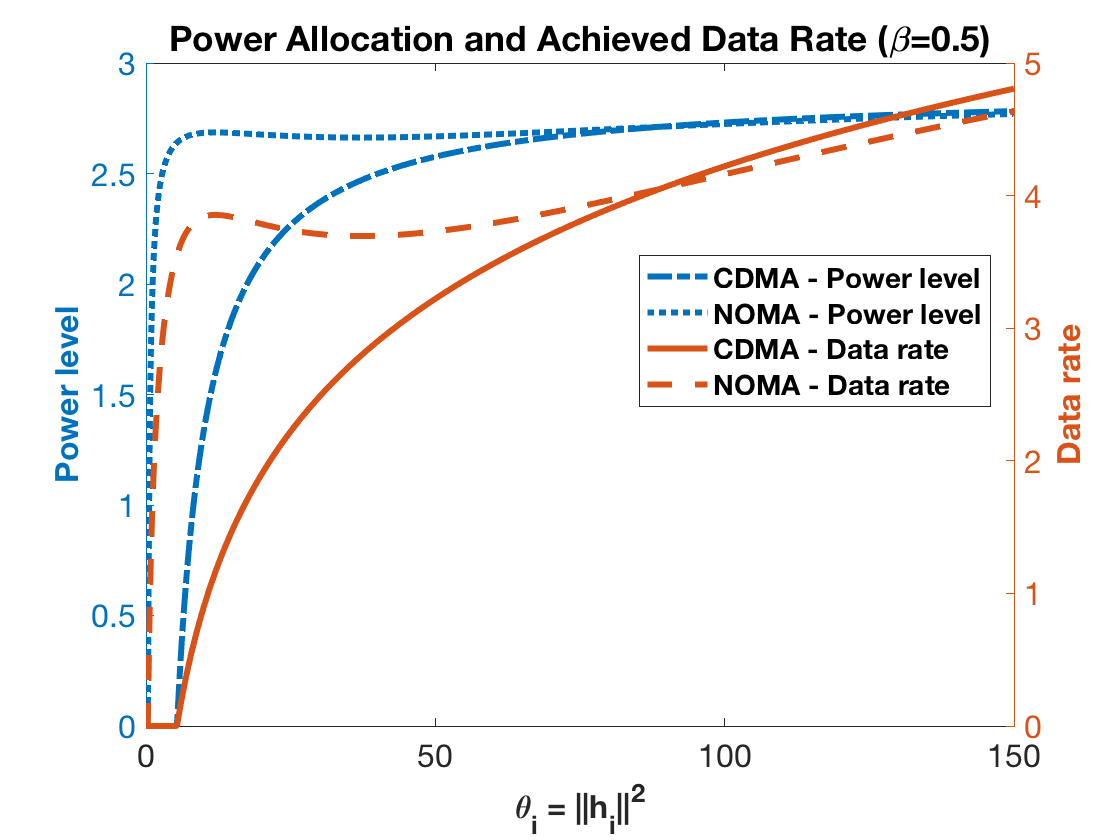}
\end{subfigure}

\begin{subfigure}{}
\includegraphics[width=0.45\textwidth]{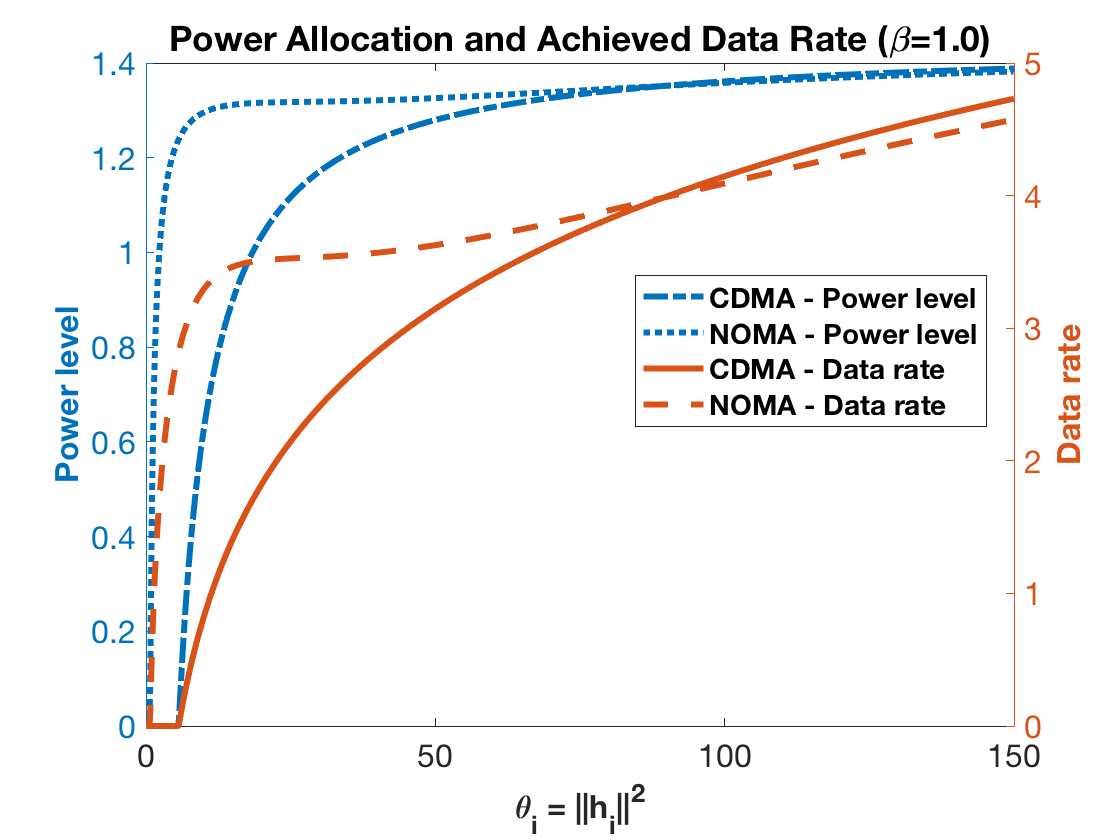}
\end{subfigure}
\begin{subfigure}{}
\includegraphics[width=0.45\textwidth]{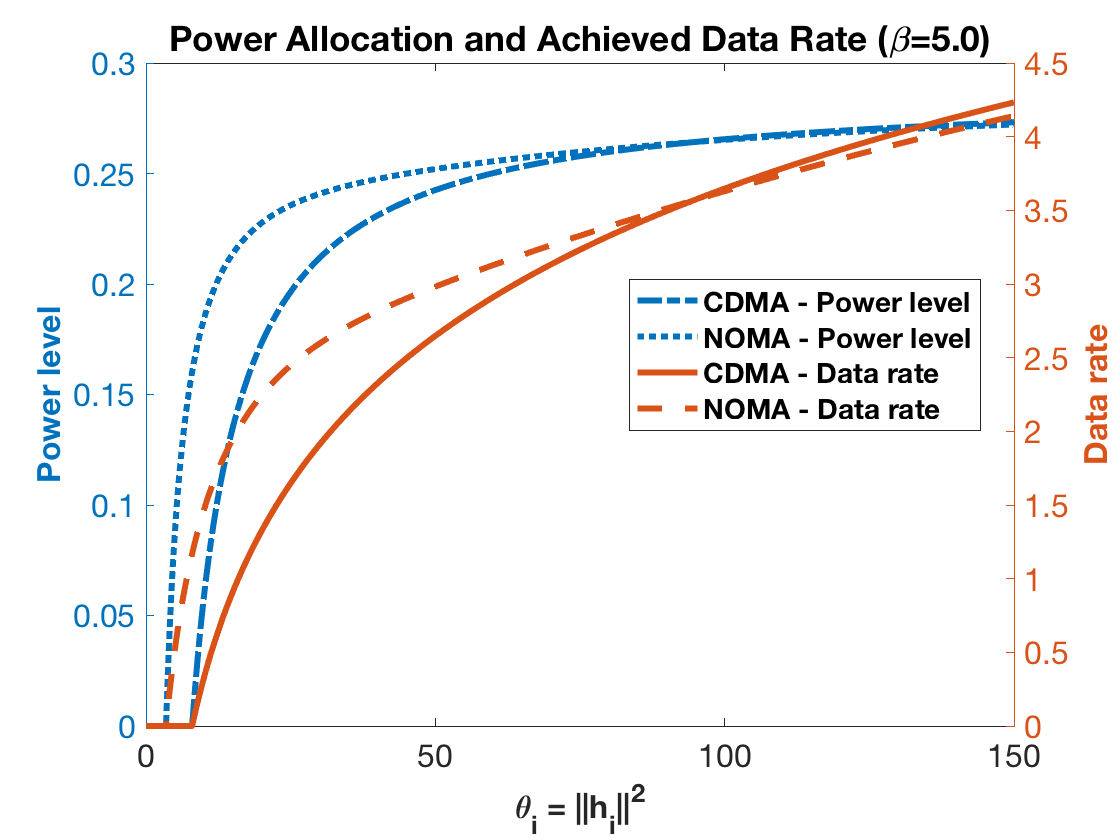}
\end{subfigure}
\caption{The equilibrium power control strategy and data rate achieved for different values of the power penalty parameter $\beta$.}
\label{fig:power_rate_compare}
\end{center}
\end{figure*}

The improvement in user fairness achieved by NOMA can be observed from the curves of achieved data rate in Fig.~\ref{fig:power_rate_compare}. With the same power penalty parameter~$\beta>0$ for unit power consumption, NOMA features a better fairness in achieved data rates than CDMA, which is discussed as follows.


For a rigorous comparison, we employ Jain's fairness index~\cite{jain1984quantitative} to evaluate the fairness in equilibrium data rate of the user population. In our simulation with $N=1000$ users, we denote the data rate achieved by user $i \ (1 \leq i \leq N)$ as $d_i \geq 0$. The Jain's fairness index is defined as
\begin{equation}
 J(d_1,d_2,\ldots,d_N) := \frac{\left( \sum\limits_{i=1}^N d_i \right)^2}{N \cdot \sum\limits_{i=1}^N d_i^2} = \frac{\overline{d}^2}{\overline{d^2}} \in (0,1],
 \end{equation}
 where a larger value of the index implies a better fairness. For the values of $\beta$ considered in Fig.~\ref{fig:power_rate_compare}, the values of Jain's fairness for CDMA and NOMA in the simulation are listed in TABLE~\ref{table:fairness}.
\begin{table}[H]
\centering
\begin{tabular}{ |c|c|c|c|c| } 
 \hline
 $\beta$ & 0.1 & 0.5 & 1.0 & 5.0 \\ \hline
 CDMA & 0.8829 & 0.8802 & 0.8768 & 0.8518\\ \hline
 NOMA & 0.9931 & 0.9864 & 0.9766 & 0.9205\\ 
 \hline
\end{tabular}
  \caption{Jain's fairness index for equilibrium data rate.}
\label{table:fairness}
\end{table}

As indicated by~Fig.~\ref{fig:power_rate_compare}, though NOMA outperforms CDMA in terms of fairness in equilibrium data rate, the increase in~$\beta$ is undesirable for NOMA due to a decreasing fairness. Aside from that, the level of achieved data rates in general only slightly decreases with a larger~$\beta$.

\begin{remark}
Based on these comparisons, we summarize some empirical findings concerning the applicability of CDMA and NOMA.

\begin{itemize}
\item[(1)] For the cases with a small cost for power consumption (i.e., $\beta>0$ takes a small value), NOMA is preferable for its advantages in the fairness achieved;
\item[(2)] For the case of costly power resources (i.e., $\beta>0$ takes a large value), the performance gap between CDMA and NOMA is negligible, while CDMA is more convenient for implementation.
\end{itemize}
\thlabel{remark:empirical}
\end{remark}

In the following, we focus on some properties theoretically shown in the main results. It is of interest to provide some numerical verification to them. In Fig.~\ref{fig:power_rate_compare}, the properties established in \thref{coro:CDMA_strategy} and \thref{coro:NOMA_strategy} concerning the continuity and monotonicity of the equilibrium strategies have already been observed. In addition, since in Fig.~\ref{fig:power_rate_compare} the equilibrium data rate curves for CDMA and NOMA intersect, the property stated in \thref{prop:cross} is also demonstrated numerically.

The comparison between equilibrium social welfare under CDMA and NOMA, as analysed in \thref{lemma:nec_optimal}, is the main results in this paper. Thus, we numerically evaluate the expected utility of all participants, i.e., the objective function $J(p,z)=\mathbb{E}[u (p(\theta_i) , p, \theta_i)]$ defined in~(\ref{eq:expected_utility}). This metric of social welfare corresponds to the average level of individual utilities among a large number of non-cooperative uplink users.

The expected utilities under the game equilibrium are evaluated for different~$\beta>0$. Moreover, the two curves in Fig.~\ref{fig:expected_utility} for CDMA and NOMA illustrate the effectiveness of NOMA in social welfare enhancement.

\begin{figure}[h!]
	\begin{center}
		\includegraphics[width=0.45\textwidth]{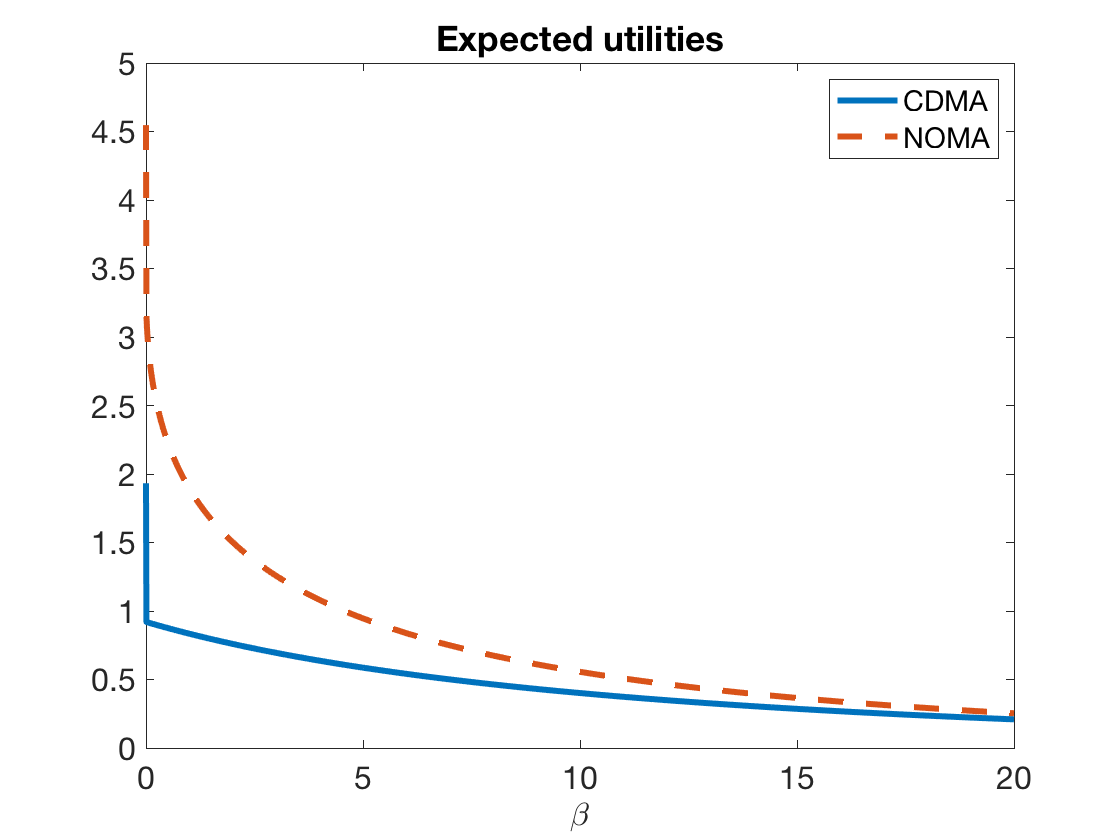}	
		\caption{The expected utility~$J(p^*,z^*)$ achieved under equilibrium strategies with CDMA and NOMA, respectively.}
		\label{fig:expected_utility}
	\end{center}
\end{figure}

It is observed from~Fig.~\ref{fig:expected_utility} that NOMA can indeed achieve a higher social welfare than CDMA, conforming the theoretical results in~\thref{lemma:nec_optimal}. In addition, we observe that when the value of $\beta>0$ is small, more effective performance improvement is achieved by NOMA, which is consistent with our intuitive analysis in \thref{remark:empirical}. When~$\beta>0$ takes a large value, the regulating effects of the energy cost dominates, and the benefits of implementing the NOMA protocol gradually shrink.

\section{Conclusion}
\label{sec:conclusions}
We have considered an uplink power control problem for wireless communication when a large number of users are competing for the channel resources. Both power-domain CDMA and NOMA are investigated. We performed equilibrium analysis of the non-cooperative channel access with an aggregative game model so that the opponents' actions are captured collectively. The existence and uniqueness of an equilibrium strategy are established for both CDMA and NOMA. Moreover, performance evaluation has been conducted under the equilibrium strategies. It turns out that NOMA achieves a better social welfare in the non-cooperative user population at its equilibrium. In addition, simulation results show an improved fairness in the equilibrium data rate under NOMA.

\appendices
\section{Proof of~\thref{THM:CDMA_contraction_mapping}}
\label{sec:proof_thm1}
\begin{proof}
For convenience of presentation, based on~(\ref{eq:BR_CDMA}), we define a ``strategy-wise" best response operator $\overline{\mathcal{BR}}: L^1(M, \mathbb{R},\nu) \to L^1(M, \mathbb{R},\nu)$ which performs updates to strategy profiles based on the utility function of all players $\theta_i \in M$. The operator~$\overline{\mathcal{BR}}$ is defined such that for any given strategy profile~$p_k \in L^1(M, \mathbb{R},\nu)$, we can obtain a new strategy profile~$\overline{p}_{k+1}$ through
\begin{equation}
\overline{p}_{k+1} := \overline{\mathcal{BR}} (p_k) \in L^1(M, \mathbb{R},\nu),
\label{eq:BR_strategy}
\end{equation}
which leads to
\begin{align*}
\overline{p}_{k+1} (\theta_i) = \frac{1}{\beta \ln 2} - \frac{\alpha \mathbb{E}[p_k (\theta_j) \theta_j] + N_0}{\theta_i} , \ \ \forall \theta_i \in M.
\end{align*}

Since in practice, the circuits can only transmit at the power levels within $\mathcal{E}$, we adopt the truncation operator $\mathcal{T}$ in \thref{def:saturation} so that a new feasible strategy $p_{k+1}$ is obtained as
\begin{equation*}
p_{k+1} = \mathcal{T} (p_{k+1}) = \mathcal{T} \circ \overline{\mathcal{BR}}(p_k) \
\in \mathcal{A}.
\end{equation*}

The normed space $\left( L^1(M, \mathbb{R},\nu), \norm{\cdot}_1^{\nu} \right)$ of functions defined on $M \subset \mathbb{R}_{++}$ is complete. The existence and uniqueness of mean-field equilibrium in the game $\mathcal{G}$ can be given in two steps.

The first step is to show that the composite operator $\mathcal{T} \circ \overline{\mathcal{BR}}$ defined on the space of feasible strategy profiles $L^1(M, \mathbb{R},\nu)$ is a contraction mapping. We pick up two arbitrary strategy profiles $p_k^{(1)}, p_k^{(2)}$ in the space $\left(L^1(M, \mathbb{R},\nu), \norm{\cdot}_1^{\nu} \right)$ and following the non-expansive results in \thref{lemma:T_nonexpansive}, we obtain
\begin{align*}
& \ \ \ \norm{ \mathcal{T} \circ \overline{\mathcal{BR}}(p_k^{(1)}) - \mathcal{T} \circ \overline{\mathcal{BR}} (p_k^{(2)})}_1^{\nu}\\
& \leq \norm{\overline{\mathcal{BR}}(p_k^{(1)}) - \overline{\mathcal{BR}} (p_k^{(2)})}_1^{\nu}\\
&=\int\limits_{\theta_i \in M} \abs{\alpha \mathbb{E}[p_k^{(1)} (\theta_j) \theta_j]- \alpha \mathbb{E}[p_k^{(2)} (\theta_j) \theta_j] } f(\theta_i) d \lambda(\theta_i)\\
&= \int\limits_{\theta_i \in M} f(\theta_i) d \lambda(\theta_i) \cdot \alpha \abs{ \mathbb{E}[p_k^{(1)} (\theta_j) \theta_j]- \mathbb{E}[p_k^{(2)} (\theta_j) \theta_j] }\\
&=P(M) \cdot \alpha \abs{\int\limits_{\theta_i \in M} [p_k^{(1)} (\theta_j) -  p_k^{(2)} (\theta_j) ] \theta_j f(\theta_j) d \lambda(\theta_j)}\\
& \leq 1 \cdot \alpha \int\limits_{\theta_i \in M} \abs{p_k^{(1)} (\theta_j) -  p_k^{(2)} (\theta_j)} \theta_j f(\theta_j) d \lambda(\theta_j)\\
&=\alpha \norm{p_k^{(1)} - p_k^{(2)}}_1^{\nu}.
\end{align*}

Since it is assumed that $\alpha < 1$, we obtain that $\alpha \in (0,1)$ based on the CDMA protocol. Therefore, the composite operator $\mathcal{T} \circ \overline{\mathcal{BR}}$ is a contraction mapping on the Banach space $\left(L^1(M, \mathbb{R},\nu), \norm{\cdot}_1^{\nu} \right)$. By Banach fixed point theorem, the operator $\mathcal{T} \circ \overline{\mathcal{BR}}$ has a unique fixed point $p^*$ in the space $\left(L^1(M, \mathbb{R},\nu), \norm{\cdot}_1^{\nu} \right)$. Moreover, in order to show that the game $\mathcal{G}$ admits a feasible mean-field equilibrium, below we show that the fixed point $p^*$ of the ``strategy-wise" best response operator $\mathcal{T} \circ \overline{\mathcal{BR}}$ lies within the space of feasible strategies~$\mathcal{A}$.

As~$\mathcal{T} \circ \overline{\mathcal{BR}}$ is a contraction mapping on the Banach space $\left(L^1(M, \mathbb{R},\nu), \norm{\cdot}_1^{\nu} \right)$, a Cauchy sequence~$\{p_k\}_{k \in \mathbb{N}}$ can be constructed with any given initial element $p_0 \in \left(L^1(M, \mathbb{R},\nu), \norm{\cdot}_1^{\nu} \right)$ such that its limit point will be exactly the unique fixed point~$p^*$ of~$\mathcal{T} \circ \overline{\mathcal{BR}}$. According to Theorem~1 in Section~2.12 of~\cite{luenberger1997optimization}, if the set of feasible strategies~$\mathcal{A}$, as a subset of the Banach space $\left(L^1(M, \mathbb{R},\nu), \norm{\cdot}_1^{\nu} \right)$, is closed, the set~$\mathcal{A}$ is complete. Hence, the limit~$p^*$ of the constructed Cauchy sequence~$\{p_k\}_{k \in \mathbb{N}}$ will lie within~$\mathcal{A}$. Now, it remains to show that~$\mathcal{A}$ is closed.

For any~$q$ in the closure~$\overline{\mathcal{A}}$ of~$\mathcal{A}$, there exists a sequence~$\{q_k\}_{k \geq 1}$ such that~$\norm{q_k - q}_1^{\nu} \to 0$ as~$k \to \infty$. According to Corollary~2.32 in~\cite{folland2013real}, this implies the existence of a subsequence~$\{q_{k_j}\}$ such that~$q_{k_j} \to q$, $\nu$-a.e. Hence, the point~$q \in \overline{\mathcal{A}}$ is~$\mathcal{E}$-valued~$\nu$-a.e, so~$q \in \mathcal{A}$. Therefore, the set~$\mathcal{A}$ is closed.

By the definition of mean-field equilibrium (\thref{def:MFE}), the existence and uniqueness of mean-field equilibrium in the game~$\mathcal{G}$ adopting CDMA are shown. Banach fixed point theorem also gives a convergent sequence to the equilibrium point, i.e., $\lim\limits_{k \to \infty} p_k = p^*$, where $p_{k+1} (\theta_i) \in \mathcal{BR}(\theta_i, p_k)$ for any~$\theta_i \in M$ and $k \geq 0$. Hence, the proof is concluded.
\end{proof}

\section{Proof of~\thref{coro:CDMA_strategy}}
\label{sec:proof_coro1}
\begin{proof}
Since the existence and uniqueness of mean-field equilibrium in the game~$\mathcal{G}$ adopting CDMA protocol have already been shown in \thref{THM:CDMA_contraction_mapping}, by the definition of a mean-field equilibrium, we can express the equilibrium strategy profile~$p^*$ as a best response, i.e.,
\begin{align*}
p^*(\theta_i) & \in \mathcal{BR} (\theta_i, p^*) =\argmax\limits_{a_i \in \mathcal{E}} \ u(a_i,p^*,\theta_i)\\
&=\argmax\limits_{a_i \in \mathcal{E}} \ \log_2 \left( 1 + \frac{\theta_i a_i}{\alpha \mathbb{E}[p^*(\theta_j) \theta_j] + N_0} \right) - \beta a_i.
\end{align*}

For an arbitrarily fixed user~$\theta_i \in M$, when the opponents' strategy is fixed to the equilibrium strategy~$p^* \in~\mathcal{A}$, it turns out that the utility~$u(a_i,p^*,\theta_i)$ is a strictly concave function with respect to~$a_i \in \mathcal{E}$. Hence, the best response of player~$\theta_i$ is a singleton, i.e.,~$\mathcal{BR} (\theta_i,p)$ takes a unique value for each~$\theta_i \in~M$ under a fixed~$p$. According to Theorem 9.17 in \cite{sundaram1996first}, for a strictly concave continuous function~$u(a_i,p,\theta_i)$ under any fixed~$p$, the single-valued maximizer~$\mathcal{BR}(\theta_i,p)$ is a continuous function with respect to the parameter~$\theta_i$. Due to the existence and uniqueness of the mean-field equilibrium strategy profile~$p^*$, the function~$p^* (\theta_i) \in~\mathcal{BR}(\theta_i,p^*)$, as a best response to itself, is a continuous function.

Next, we proceed to show that the equilibrium strategy profile $p^*: M \to \mathcal{E}$ is monotonically increasing with respect to $\theta_i \in M$. Beforehand, it is necessary to show that the utility function $u(a_i,p,\theta_i)$, for any fixed $p$, satisfies strictly increasing difference in $(a_i,\theta_i)$. In other words, we need to verify that
\begin{equation}
u(a_i^+,p,\theta_i^+) - u(a_i^-,p,\theta_i^+) > u(a_i^+,p,\theta_i^-) - u(a_i^-,p,\theta_i^-)
\label{ineq:increasing_difference}
\end{equation}
for any $a_i^+ > a_i^-$ and $\theta_i^+ > \theta_i^-$ given a fixed $p$.

From the expression of the utility function $u$, we obtain
\begin{align*}
& \ \ \ \ u(a_i^+,p,\theta) - u(a_i^-,p,\theta)\\
&= \log_2 \left( 1 + \frac{a_i^+ - a_i^-}{\frac{\alpha \mathbb{E}[p^*(\theta_j) \theta_j] + N_0}{\theta} + a_i^-} \right) - \beta (a_i^+ - a_i^-).
\end{align*}
Thus, it is obvious that $u(a_i^+,p,\theta) - u(a_i^-,p,\theta)$ is monotonically increasing with respect to $\theta$ for any given $a_i^+ > a_i^-$, which leads to~(\ref{ineq:increasing_difference}). Besides, for any fixed $p$ and $\theta_i$, the utility function $u$ is a continuous function defined on a compact interval $\mathcal{E}=[u_{\min},u_{\max}]$. According to the extreme value theorem, the utility $u$ must attain its maximum within $\mathcal{E}$ for any given $p$ and $\theta_i$. Therefore, according to Theorem 10.6 in \cite{sundaram1996first}, the mean-field equilibrium strategy $p^*$, as the optimal action for maximizing $u(a_i,p^*,\theta_i)$, is monotonically increasing with respect to the identifier $\theta_i \in M$.
\end{proof}

\section{Proof of~\thref{THM:NOMA2_contraction_mapping}}
\label{sec:proof_thm2}
\begin{proof}
Similar to the analysis in the CDMA case, based on~(\ref{eq:BR_NOMA2}), we can define a ``strategy-wise" best response operator $\overline{\mathcal{BR}}_{\text{ordered}}: L^{\infty}(M, \mathbb{R},\nu) \to L^{\infty}(M, \mathbb{R},\nu)$. The operator~$\overline{\mathcal{BR}}_{\text{ordered}}$ is defined such that for any given strategy profile $p_k \in L^{\infty}(M, \mathbb{R},\nu)$, we can obtain a new strategy profile $\overline{p}_{k+1}$ as the optimal response to $p_k$ through
\begin{equation}
\overline{p}_{k+1} := \overline{\mathcal{BR}}_{\text{ordered}} (p_k) \in L^{\infty}(M, \mathbb{R},\nu),
\label{eq:BR_strategy}
\end{equation}
which gives
\begin{align*}
\overline{p}_{k+1} (\theta_i) = \frac{1}{\beta \ln 2} - \frac{\alpha \mathbb{E}[p_k (\theta_j) \theta_j \bm{1}_{\{\theta_j < \theta_i\}}] + N_0}{\theta_i} , \ \ \forall \theta_i \in M.
\end{align*}

To take into account the power constraints, we define a truncation operator $\tilde{\mathcal{T}}: L^{\infty}(M, \mathbb{R},\nu) \to L^{\infty}(M, \mathbb{R},\nu)$ in analogy to \thref{def:saturation}. The operator $\tilde{\mathcal{T}}$ is defined such that $\tilde{p} = \tilde{\mathcal{T}}(p)$ if and only if $\tilde{p}(x) = P_{\mathcal{E}} (p(x))$ for any $x \in M$. Then, with similar arguments as in \thref{lemma:T_nonexpansive}, the operator $\tilde{\mathcal{T}}$ is also non-expansive on the Banach space $\left( L^{\infty}(M, \mathbb{R},\nu), \norm{\cdot}_{\infty}^{\nu} \right)$.

We adopt the truncation operator $\tilde{\mathcal{T}}$ so that a new feasible strategy $p_{k+1}$ is obtained as
\begin{equation*}
p_{k+1} = \tilde{\mathcal{T}} (\overline{p}_{k+1}) = \tilde{\mathcal{T}} \circ \overline{\mathcal{BR}}_{\text{ordered}}(p_k) \in \mathcal{A}.
\end{equation*}

Since according to~\thref{def:MFE}, the definition of a mean-field equilibrium is equivalent to the fixed point of the best response operator, the uniqueness of the fixed point solution implies the uniqueness of a mean-field equilibrium of the game~$\tilde{G}$. By Banach fixed-point theorem, there is a unique solution $p^*_{\text{ordered}} \in \mathcal{B}$ to the fixed point equation $p = \tilde{\mathcal{T}} \circ \overline{\mathcal{BR}}_{\text{ordered}} (p)$ if $\tilde{\mathcal{T}} \circ \overline{\mathcal{BR}}_{\text{ordered}}$ is a contraction mapping on $L^{\infty}(M, \mathbb{R},\nu)$.

We pick up two arbitrary transmission power control policies $p_k, \ p'_k \in L^{\infty}(M, \mathbb{R},\nu)$ and obtain
\begin{align*}
& \ \ \ \norm{\tilde{\mathcal{T}} \circ \overline{\mathcal{BR}}_{\text{ordered}} (p_k) - \tilde{\mathcal{T}} \circ \overline{\mathcal{BR}}_{\text{ordered}} (p'_k)}_{\infty}^{\nu}\\
& \leq \norm{\overline{\mathcal{BR}}_{\text{ordered}} (p_k) - \overline{\mathcal{BR}}_{\text{ordered}} (p'_k)}_{\infty}^{\nu}\\
& \leq \norm{\frac{\alpha \int\limits_{M \cap (0,\theta_i]} \abs{p_k(y) - p'_k(y)} y f (y) dy}{\theta_i}}_{\infty}^{\nu}\\
& \leq \norm{\frac{\alpha \int\limits_{M \cap (0,\theta_i]}  \norm{p_k - p'_k}_{\infty}^{\nu} y f (y) dy}{\theta_i}}_{\infty}^{\nu}\\
&=  \norm{p_k - p'_k}_{\infty}^{\nu} \cdot \norm{ \frac{\alpha \int\limits_{M \cap (0,\theta_i]} y f (y) dy}{\theta_i}}_{\infty}^{\nu}\\
&\leq \sup\limits_{\theta_i \in M} \frac{\alpha \int\limits_{M \cap (0,\theta_i]} y f (y) dy}{\theta_i} \cdot \norm{p_k - p'_k}_{\infty}^{\nu}\\
&\leq \sup\limits_{\theta_i \in M} \frac{\alpha \int\limits_{M \cap (0,\theta_i]} \theta_i f(y) dy}{\theta_i} \cdot \norm{p_k - p'_k}_{\infty}^{\nu}\\
&= \sup\limits_{\theta_i \in M} \ \alpha \int\limits_{M \cap (0,\theta_i]} f(y) dy \cdot \norm{p_k - p'_k}_{\infty}^{\nu} \leq \alpha \cdot \norm{p_k - p'_k}_{\infty}^{\nu}.
\end{align*}

To summarize, we have shown that the inequality $\norm{\tilde{\mathcal{T}} \circ \overline{\mathcal{BR}}_{\text{ordered}} (p_k) - \tilde{\mathcal{T}} \circ \overline{\mathcal{BR}}_{\text{ordered}} (p'_k)}_{\infty}^{\nu} \leq \alpha \norm{p_k - p'_k}_{\infty}^{\nu}$ holds, where~$0<\alpha<1$. Therefore, the composite operator~$\tilde{\mathcal{T}} \circ \overline{\mathcal{BR}}_{\text{ordered}}$ is a contraction mapping on the Banach space~$\left( L^{\infty}(M, \mathbb{R},\nu), \norm{\cdot}_{\infty}^{\nu} \right)$. Similar to \thref{THM:CDMA_contraction_mapping}, in order to show that the fixed point~$p^*_{\text{ordered}}$ reside in~$\mathcal{A}$, it remains to show that the set of feasible strategies $\mathcal{A}$ is closed.

Again, we pick up a point~$q$ in the closure~$\overline{\mathcal{A}}$ of~$\mathcal{A}$. Then there exists a sequence~$\{q_k\}_{k \geq 1}$ such that~$\norm{q_k - q}_{\infty}^{\nu} \to 0$ as~$k \to \infty$. Hence, according to~\cite{folland2013real}, a subsequence~$\{q_{k_j}\}$ exists such that~$q_{k_j} \to q$, $\nu$-a.e. Then~$q$ takes value in~$\mathcal{E}$, $\nu$-a.e. In other words, we have~$q \in \mathcal{A}$.

Therefore, it follows that there is a unique mean-field equilibrium~$p^*_{\text{ordered}}$ when NOMA is adopted. Banach fixed point theorem also gives a convergent sequence to the equilibrium point, i.e., $\lim\limits_{k \to \infty} p_k = p^*_{\text{ordered}}$, where $p_{k+1} (\theta_i) \in \mathcal{BR}_{\text{ordered}} (\theta_i, p_k)$ for any~$\theta_i \in M$ and $k \geq 0$. This completes the proof.
\end{proof}

\section{Proof of~\thref{coro:NOMA_strategy}}
\label{sec:proof_coro2}
\begin{proof}
By definition, the equilibrium strategy $p^*_{\text{ordered}}$ is the best response to itself given the utility function $u(a_i, p, \theta_i)$ for NOMA. In \thref{THM:NOMA2_contraction_mapping}, we have already established the existence and uniqueness of $p^*_{\text{ordered}}: M \to \mathcal{E}$. Then, we have
\begin{align*}
p^*_{\text{ordered}} (\theta_i) &\in \mathcal{BR}_{\text{ordered}} (\theta_i,p^*_{\text{ordered}})\\
& =: \argmax\limits_{a_i \in \mathcal{E}} \ u(a_i, p^*_{\text{ordered}}, \theta_i),
\end{align*}
where the last term is a singleton for each given $\theta_i \in M$.

By Theorem 9.17 in~\cite{sundaram1996first}, the best response $\mathcal{BR}_{\text{ordered}} (\theta_i,p^*_{\text{ordered}})$ is a upper semi-continuous correspondence on $M$. As $\mathcal{BR}_{\text{ordered}} (\theta_i,p^*_{\text{ordered}})$ is a singleton for any $\theta_i \in M$, it can be concluded that $p^*_{\text{ordered}} (\theta_i) \in \mathcal{BR}_{\text{ordered}} (\theta_i,p^*_{\text{ordered}})$ is continuous with respect to~$\theta_i \in M$.
\end{proof}

\bibliographystyle{IEEEtran}
\bibliography{references.bib}

\newpage
\begin{IEEEbiography}
{Yuchi Wu} received the B.Eng. degree from School of Electronic and Information Engineering, Beihang University, Beijing, China, in 2015. Later, he received the M.Phil. degree in 2017 and the Ph.D. degree in 2020, both from Department of Electronic and Computer Engineering, Hong Kong University of Science and Technology, Kowloon, Hong Kong. From March 2020 to May 2020, he was a visiting student in the School of Mathematics and Statistics, Carleton University, Ottawa, ON, Canada. His research interests include mean-field games, large-scale systems and state estimation in wireless sensor networks.
\end{IEEEbiography}

\begin{IEEEbiography}
{Junfeng Wu} received the B.Eng. degree from the
Department of Automatic Control, Zhejiang University, Hangzhou, China, and the Ph.D. degree in Electronic and Computer Engineering from Hong Kong
University of Science and Technology, Hong Kong,
in 2009, and 2013, respectively. From September
to December 2013, he was a Research Associate
in the Department of Electronic and Computer Engineering, Hong Kong University of Science and
Technology. From January 2014 to June 2017, he
was a Postdoctoral Researcher in the ACCESS (Autonomic Complex Communication nEtworks, Signals and Systems) Linnaeus
Center, School of Electrical Engineering, KTH Royal Institute of Technology,
Stockholm, Sweden. He is currently with the College of Control Science and
Engineering, Zhejiang University, Hangzhou, China. His research interests
include networked control systems, state estimation, and wireless sensor
networks, multi-agent systems. Dr. Wu received the Guan Zhao-Zhi Best Paper
Award at the 34th Chinese Control Conference in 2015.  He served as an associate editor for the European Control Conference 2020 and  a technical associate editor for the IFAC World Congress 2020.
\end{IEEEbiography}

\begin{IEEEbiography}
{Minyi Huang} received the B.Sc. degree from Shandong
University, Jinan, Shandong, China, in 1995, the M.Sc.
degree from the Institute of Systems Science, Chinese
Academy of Sciences, Beijing, in 1998, and the Ph.D. degree
from the Department of Electrical and Computer Engineering,
McGill University, Montreal, QC, Canada, in 2003,
all in systems and control.

He was a Research Fellow first in the Department of
Electrical and Electronic Engineering, the University of
Melbourne, Melbourne, Australia, from February 2004 to
March 2006, and then in the Department of Information
Engineering, Research School of Information Sciences and Engineering, the Australian
National University, Canberra, from April 2006 to June 2007. He joined the
School of Mathematics and Statistics, Carleton University, Ottawa, ON, Canada as
an Assistant Professor in 2007, where he is now a Professor. His research interests
include mean field stochastic control and dynamic games, multi-agent control and
computation in distributed networks with applications.
\end{IEEEbiography}

\begin{IEEEbiography}
{Ling Shi} received the B.S. degree in electrical and electronic engineering from Hong Kong University of Science and Technology, Kowloon, Hong Kong, in 2002 and the Ph.D. degree in Control and Dynamical Systems from California Institute of Technology, Pasadena, CA, USA, in 2008. He is currently a Professor in the Department of Electronic and Computer Engineering, and the associate director of the Robotics Institute, both at the Hong Kong University of Science and Technology. His research interests include cyber-physical systems security, networked control systems, sensor scheduling, event-based state estimation and exoskeleton robots. He is a senior member of IEEE. He served as an editorial board member for the European Control Conference 2013-2016. He was a subject editor for International Journal of Robust and Nonlinear Control (2015-2017). He has been serving as an associate editor for IEEE Transactions on Control of Network Systems from July 2016, and an associate editor for IEEE Control Systems Letters from Feb 2017. He also served as an associate editor for a special issue on Secure Control of Cyber Physical Systems in the IEEE Transactions on Control of Network Systems in 2015-2017. He served as the General Chair of the 23rd International Symposium on Mathematical Theory of Networks and Systems (MTNS 2018). He is a member of the Young Scientists Class 2020 of the World Economic Forum (WEF).
\end{IEEEbiography}
\end{document}